\documentclass[reqno,11pt]{amsart}
\usepackage{stmaryrd}

\usepackage{amsmath}
\usepackage{amsfonts}
\usepackage{amssymb}
\usepackage{amsthm}

\usepackage{physics}
\usepackage{tensor}
\usepackage{slashed}

\usepackage{bbold}
\usepackage{mathtools}

\usepackage[nolist]{acronym}
\usepackage{hyperref}

\usepackage{biblatex}

\usepackage{pgfplots}

\usepackage{color} 
\usepackage{soul} 
\DeclareMathOperator{\Lin}{L}

\DeclareMathOperator{\Span}{span}
\DeclareMathOperator{\supp}{supp}
\DeclareMathOperator{\id}{id}

\DeclareMathOperator{\Var}{Var}

\newcommand{\defeq}{\coloneqq}
\newcommand{\eqdef}{\eqqcolon}

\newcommand{\diracad}[1]{\overline{#1}}
\newcommand{\complexconj}[1]{\overline{#1}}

\newcommand{\opnorm}[1]{
        {\left| \kern-0.25ex \left| \kern-0.25ex \left| #1 \right| \kern-0.25ex \right| \kern-0.25ex \right|}
}

\newcommand{\spinbra}[1]{\left. \prec #1 \right|}
\newcommand{\spinket}[1]{\left| #1 \succ \right.}
\newcommand{\spinbraket}[2]{\prec #1 | #2 \succ}
\newcommand{\Freg}{\F_n^{\text{reg}}}

\newcommand{\N}{\mathbb{N}}
\newcommand{\R}{\mathbb{R}}
\newcommand{\C}{\mathbb{C}}
\renewcommand{\H}{\mathcal{H}}
\newcommand{\F}{{\mathcal{F}}}
\newcommand{\s}{{\mathfrak{s}}}
\renewcommand{\L}{{\mathcal{L}}}

\newtheorem{definition}{Definition}[section]
\newtheorem{proposition}[definition]{Proposition}
\newtheorem{lemma}[definition]{Lemma}
\newtheorem{corollary}[definition]{Corollary}

\newtheorem{conjecture}[definition]{Conjecture}

\numberwithin{equation}{section}

\begin{acronym}
    \acro{cfs}[CFS]{causal fermion system}
    \acro{el}[EL]{Euler-Lagrange}
    \acro{qt}[QT]{quantum theory}
    \acro{qrf}[QRF]{quantum reference frames}
    \acro{gr}[GR]{general relativity}
    \acro{fcdm}[FCDM]{fermionic condensate dark matter}
\end{acronym}

\begin{document}
    \begin{abstract}
    In this paper, we argue that spacetime in causal fermion systems can be understood as the web of correlations of a many-body quantum system.
    This argument highlights the fact that causal fermion systems is a completely relational theory.
    We also explain how our perception of a background (spacetime) emerges in the limit where the number of states taken in consideration goes to infinity. This thereby constitutes a complete viable ontology for causal fermion systems which are not reliant on the continuum limit. A key insight is the fact that in a relevant subset of causal fermion systems, which includes the continuum limit of the Minkowski vacuum spacetime,  minimization of the causal action can be understood as a minimization of fluctuations in the causal structure of spacetime. 
\end{abstract}

    \title[Causal Fermion Systems]{Causal Fermion Systems: Spacetime as the web of correlations of a many-body quantum system}

    \author[P. Fischer]{Patrick Fischer$^*$}
    \address{$^*$ Fakultät für Mathematik \\ Universität Regensburg \\ D-93040 Regensburg \\ Germany}
    \email{patrick.fischer@ur.de}

    \author[C. F. Paganini]{Claudio F. Paganini$^{*, \dagger}$}
   
    \address{$^\dagger$ Fakultät für Mathematik \\ TU Chemnitz\\ D-09111 Chemnitz \\ Germany}
    \email{claudio.paganini@mathematik.tu-chemnitz.de}

    \maketitle

    \tableofcontents

    \section{Introduction}\label{sec:introduction}

In recent years (co)relational approaches to fundamental physics have been gaining popularity.
Be this Rovelli's relational interpretation of quantum theory~\cite{rovelli1996relational,martin2023fact, smerlak2007relational,laudisa2019open,yang2019relational,robson2024relational}, the idea of relative locality by Amelino-Camelia et al.~\cite{amelino2011principle}, or the insight by Kempf et al.~that the metric is
encoded in the coincidence limit of the two-point correlation function~\cite{kempf2021correlational, reitz2023model}.
Correlations also play a central role in several other works~\cite{ormrod2024quantum,smolin2020temporal,oldofredi2021bundle,vidotto2022relational,kastner2022transactional, padmanabhan2020probing} with e.g.~Jia~\cite{jia2021correlationalworld,jia2022complexsimplicalQG,jiatimeorder} putting an emphasis on causal
correlation via Synge's world function.
Correlations also notably play a key role in our observations of the early universe, namely the fluctuations of the cosmic
microwave background~\cite{planck2016astronomy}.

Meanwhile, in the past two decades, Finster developed the theory of \ac{cfs}~\cite{cfs, website} with a very different
idea in mind, namely that the Dirac sea plays an essential role.
In terms of the mathematical results, \ac{cfs} is by now one of the leading contenders for a theory that unifies general relativity
with quantum theory.
The theory has not only succeeded in deriving general relativity and the standard model (gauge groups and their correct
representations) in suitable limiting cases, but in doing so has provided us with a resolution of two long-standing
puzzles in fundamental physics: the three generations of fermions and the relative weakness of gravity~\cite{cfs}.
The recently completed derivation of quantum field theory~\cite{qft} completes the reconstruction of our current best
theories within the appropriate limits.
Recently, first proposals towards a distinguishing phenomenology have been put forwards, regarding e.g.~baryogenesis~\cite{baryogenesis}
and the resolution of the measurement problem via an effective collapse model~\cite{finster2024causal}.

However, despite these successes, so far we only have a sense of a physical interpretation, i.e.~an operationalization
of the fundamental mathematical objects of the theory in terms of concepts we have experience with in the world around us,
in the continuum limit, where we can rely on an effective description in terms of the well established notions of spacetime and matter fields.
In this paper, we close this gap and put forward a first proposal for a full physical interpretation of \ac{cfs} that is a priori independent of the effective description in terms of classical spacetime and matter fields.
We propose a concise and viable ontology of the theory, starting with the realization that the fermionic projector, which has
been introduced by Finster to encapsulate the information contained in the Dirac sea; is, in Minkowski space, nothing else
than the fermionic two-point correlator.
Therefore, it turns out that \ac{cfs} is a possible realization of the ideas put forward by Kempf et al.~\cite{kempf2021correlational,reitz2023model} that the
fundamental properties of the universe should be encoded by correlations alone.

The fundamental ontology proposed in this paper integrates concepts and analogies from various fields, tailoring them to
the requirements of \ac{cfs}.
First, from information theory, the Shannon's sampling theorem, which is a classical result that establishes a
sufficient condition for a sample rate that permits a discrete sequence of samples to capture all the information from a
continuous-time signal of finite bandwidth.
The intuition we want to borrow here is that in the beginning the approximation improves with every added sample, while
at some point the approximation is exact and additional samples add no further information.
This result provides the intuition for reconstructing information in \ac{cfs} by adding more and more wave functions.
This also runs into a limit once we exhaust the space of occupied states in our physical system.
At this point, we argue that to make contact with conventional formulations of $N$-body systems, it is useful to add further
states to this system.
However, in the spirit of Shannon's sampling theorem, these auxiliary states do not contain additional information about
the physical system.

Second, we recall that it is a standard feature of many theories that one imagines the behavior of a test object to trace
out the properties of the physical system of interest.
General relativity takes this one step further, where one studies families of test particles to probe the structure of
spacetime via the geodesic deviation equation or the Jacobi equation.
Now the key difference here is that mathematically the test particle and the physical system it tests are different objects,
i.e.~geodesics vs.~manifolds with metric structures.
It is implicitly understood that in the real world, the test bodies are, of course, also part of the system and influence spacetime.
For experiments, this is usually a good enough approximation as long as the test particle is sufficiently small relative
to the physical system being tested\footnote{
    In general relativity the effect of the test mass on the trajectory is a subject of intense study and goes under the
    name of the self-force problem, see, e.g.,~\cite{harte2018foundations,thompson2019gravitational} and references therein.
}.
Now in \ac{cfs} the difference is that, in principle, we can pick any sufficiently small subsystem (i.e.~a collection of states)
as a test system and treat the remaining states as the background whose structure is to be tested.
Here, the Hilbert space containing the aforementioned auxiliary states comes in handy because we might as well treat
suitably chosen states as test systems, much in the spirit of geodesics in general relativity.

Our next observation provides contact with the ideas developed in the framework of \ac{qrf}~\cite{lake2023quantum,goeller2022frames, quantumFrameRelativity, bartlett2007reference, fewster2025quantum,Apadula2024quantumreference, ali2024quantum, Krumm2021quantumreference,Yang2020switchingquantum,carette2025operational,giacomini2019quantum}.
Here, it is important to note that there is one caveat to the above statement that any subsystem can serve as a test system:
it is intuitively clear that in order for an experiment to be conducted, a system needs to be suitably localized.
Therefore, not every split of the system into background and probe is equally viable as a description for laboratory experiments.
To clarify which splits are viable, we introduce the concepts of localized and delocalized states in Section~\ref{subsec:spacetime-superposition} as well as the idea
of quasilocalized states in the outlook.
The different viable choices of a split (more formally a representation of the subspace of the background space)
can be understood as a change in reference frames.
In this context, we pay particular attention to the fact of how the -- a priori ad hoc -- split into background and probe enables
us to explain the emergence of spacetime as a background.
On this background, we study fields and particles from within the fundamental structures of \ac{cfs} alone.
This is primarily discussed in Section~\ref{sec:experiment} where we elaborate on how habitually we perform this split
between background and probe when setting up any experiment.
This way of constructing a background/ (quantum) reference frame is strictly incompatible with some of the scenarios studied
in the context of \ac{qrf} as we explain in this section.

Fourth, we want to incorporate the intuition underlying many particle quantum systems~\cite{fetter2012quantum,zagoskin2016quantum,wen2004quantum} into the description of \ac{cfs}.
In contrast to the standard formulation of \ac{cfs} we work with the larger auxiliary Hilbert space together with a state
represented by a density matrix that projects on the occupied states.
Although this carries an extra baggage of redundant information, we believe that this picture is beneficial to understand the full fundamental ontology proposed in this paper.
This is another case where working with an auxiliary Hilbert space makes it easier to connect to established frameworks.

Finally, based on Kempf et al.~\cite{kempf2021correlational,goeller2022frames}, we assume a correspondence between the correlation strength and the causal distance.
In the spirit of the approximation procedure in Shannon's sampling theorem throughout this paper, we follow the path starting from
individual states and build the entire system from there by summing over all states.
In that spirit, we first introduce a set of observables for the individual particles in the many-body system,
which we then aggregate to observables for the system as a whole.
The study of these observables was originally motivated by the arguments in~\cite{dzhunushaliev2023quantum,dzhunushaliev2023quantization}
to consider theories with operator-valued measures.
These ideas have a natural realization in the context of \ac{cfs}, because the Lagrangian of the theory is formulated in
terms of a measure over a subset of the bounded linear operators.
Under suitable assumptions, these operators possess a manifold structure~\cite{Finster2021}, thus merging the mathematical structures
underlying \ac{gr} and \ac{qt}.

An important point to note here is that while summing over the contributions of individual states works for all the properties
of the system mentioned so far, it fails to hold for the action at the heart of \ac{cfs}, because it turns out to be nonlinear
in the states.
This nonlinearity can be easily understood in the context of an interesting subset of \ac{cfs} for which we prove that
the causal action principle can be written as the variance of the two-point correlation strength.
Minimizing the action therefore leads to the principle of \textit{minimal fluctuations}.

In summary, the ontology for \ac{cfs} that we present here states that the information of a system is encoded entirely in the two-point correlation strength and the
position observables of the states in the Hilbert space. 
We conclude that all we can ever hope to measure in an experiment is the (spacetime) position of fermions, these are the
`beables' of the theory\footnote{It has already been argues in the context of Bohmian mechanics that fundamentally we can only really measure the positions of particles\cite{esfeld2014ontology,gisin2018bohmian,lazarovici2020position}. Note that this is in stark contrast to the transactional interpretation of quantum mechanics~\cite{kastner2022transactional} that holds, that momenta are the fundamental observables.}.
All other physical concepts that we perceive, including all the gauge fields, are emergent observables inferred from the
change in position of fermions.
These emergent concepts are chosen in such a way that they enable the simplest possible description of the spacetime dynamics
of the observed fermions.

The organization of the paper is as follows.
In Section~\ref{sec:preliminaries} the structures of \ac{cfs} are briefly introduced.
In Section~\ref{sec:ontology-of-cfs} we discuss the observables for subsystems in the many-body quantum system, as well as the
observables of the total system.
In passing, we introduce the idea of a single particle sub-spacetime and its superposition.
Additionally, we present a suggestive derviation of the closed chain and there kernel of the fermionic projector based
on considerations of causal correlations between spacetime points.
In Section~\ref{sec:continuum-limit} the continuum limit of \ac{cfs} is discussed for the example of Minkowski spacetime.
Thereby, we show that the fermionic projector matches the two-point correlator familiar from standard formulations of quantum field theory in Minkowski space.
Moreover, in this limit, we see that we can cast the Lagrangian in the causal action principle into a more suggestive form, namely as the variance of the two-point correlation strength. 
Finally, in Section~\ref{sec:experiment} we provide a full description of an abstract experiment relying purely on the structures of the theory itself.
Finally, in section~\ref{sec:outlook} we discuss a number of important questions raised by the broad considerations in this paper.

\section{Preliminaries: Causal Fermion Systems} \label{sec:preliminaries}

The established definition of a \ac{cfs} consists of three objects:
\begin{enumerate}
    \item a Hilbert space~$\H$,
    \item a suitably chosen subset~$\F_n$ of the bounded linear operators on the Hilbert space,
    \item a measure~$\rho$ defined on the Borel $\sigma$-algebra with respect to the norm-topology on the bounded linear
    operators which we denote by~$\Lin(\H)$.
\end{enumerate}

Formally, the full definition is given below.
\begin{definition}[Causal Fermion System]
    \label{def:cfs}
    Let $(\H, \braket{\cdot}{\cdot})$ be a separable complex Hilbert space, $n \in \N$
    and $\F_n \subset \Lin(\H)$ be the set of all symmetric operators on $\H$ of finite rank, which (counting multiplicities) have at most $n$ positive and at most $n$ negative eigenvalues.
    Further, let $\rho$ be a positive measure (defined on a $\sigma$-algebra of $\F_n$).
    We refer to $(\H, \F_n, \rho)$ as a \textbf{causal fermion system} of \textbf{spin dimension} $n$.
\end{definition}

The operators with exactly~$n$ positive and~$n$ negative eigenvalues constitute the subset~$\F_n^{\text{reg}}$ of
\textit{regular points} in the set~$\F_n$.
One can show~\cite{Finster2020, Finster2021} that the subset~$\F_n^{\text{reg}}$ is a Banach manifold.
Hence, the theory of \ac{cfs} is based on the mathematical structure of an \textit{operator manifold}, thereby fusing the
underlying structure of \ac{gr} and quantum theory.
An element $x$ in $\F_n$ encodes the correlations of all states in the Hilbert space at a point.
we therefore identify this correlation information with a potential point in spacetime.
The measure $\rho$ then determines which correlations are realized in a particular physical system of spin dimension $n$.

Near the continuum limit, we can make this identification of a spacetime point with the correlation information more concrete.
To this end, we introduce a map from a classical spacetime\footnote{
    Throughout the paper classical spacetime refers to a $(3 + 1)$-dimensional, time-orientable, Lorentzian manifold while
    spacetime corresponds to the structure in \ac{cfs} defined below.
} and matter configuration into the operator manifold~$\F_n^{\text{reg}}$.
At this point we only introduce the so-called \textit{local correlation map} $F$ as an abstract concept
\begin{align}\label{eq:localcorrelation}
    F[g_{\mu\nu}, A_\mu,\dots]: \mathcal{M}  \qquad &\to \qquad \qquad\F_n \\
    x \qquad &\mapsto \qquad F[g_{\mu\nu}, A_\mu,\dots](x).
\end{align}
For a more detailed motivation of this construction see e.g.~\cite{mmt-cfs}.
For our purpose, the crucial point is that the map allows us to identify points in a classical spacetime~$\mathcal{M}$
with operators in the set~$\F_n$.
This motivates our present attempt to formulate a direct ontology for \ac{cfs} that is independent of the continuum limit
(although, as we discuss in Section~\ref{sec:experiment} to interpret experiments, the continuum limit is useful nevertheless).
We explain the construction of this map in Minkowski space in some detail in Section~\ref{sec:causalcorrelations} while
the precise construction can be found in~\cite{cfs}.
The map $F$ depends on the metric $g_{\mu\nu}$, as well as the matter fields defined on the classical spacetime under consideration.
Then, the local correlation map defines a measure $\rho$ of a subset $\Omega \subset \F_n$ as the push forward of the
measure $\mu$ of the classical spacetime under the map $F$
\begin{align}
    \label{eq:pushforward}
    \rho(\Omega) \defeq \mu(F^{-1}(\Omega))= \int_{F^{-1}(\Omega)} \Phi.
\end{align}
Here~$\Phi$ is the volume form of the measure on the classical spacetime.

If we consider the metric measure on the manifold, then $\Phi = \sqrt{\abs{g}}\dd^4 x$.
However, the construction of the push-forward measure also works for more general measures on the manifold $\mathcal{M}$~\cite{mmt-cfs}.
The idea we want to convey here is that the fundamental physical reality is made up of a web of correlations encoded in
the operators $x \in \F$ and that the local correlation map provides us with a dictionary to translate the structure of
this complicated web of correlations into the simple language of classical spacetime endowed with suitable matter fields.
These conventional concepts are an emergent effective description of the physical system, much in the same spirit as
e.g.~temperature and pressure arise in thermodynamics.

For notational simplicity, we denote the image of a classical spacetime point $x$ under the local correlation map
$F\,[g_{\mu\nu}, A_\mu,\dots](x) \defeq x$ with the same symbol.
That means that depending on the context $x$ refers to a point in classical spacetime or to the operator that associates
with it under the local correlation map.
In the same spirit, we also refer to the operators as points in spacetime.
This serves the main purpose of making clear that the operators in $\F_n$ should ultimately be thought of as points in spacetime.

With the local correlation map at hand, we introduce the notion of ``spacetime'' in the setting of \ac{cfs} as the support\footnote{
    The \textit{support} of a measure is defined as the complement of the largest open set of measure zero, i.e.
    \begin{align*}
        \supp \rho \defeq \F_n \setminus \bigcup \big\{ \text{$\Omega \subset \F_n$ \,\big|\,$\Omega$ is open and~$\rho(\Omega)=0$} \big\} \:.
    \end{align*}
}
of the measure,
\begin{align*}
    M \defeq \supp \rho.
\end{align*}
Hence, in the context of \ac{cfs} the measure is the fundamental variable and its support is typically only a proper subset
of the set~$\F_n$.
This contrasts to \ac{gr} where the support of the measure is the entire manifold on which the variation principle is
defined.
In particular, the dimension of the spacetime in \ac{cfs} can be much lower than the dimension of the set~$\F_n$
(or the manifold~$\F_n^{\text{reg}}$ for that matter) which in principle can be infinite.
Furthermore, despite the subset~$\F_n^{\text{reg}}$ being a manifold, the spacetime $M$ need not be so, but it can in fact
also have a \textit{discrete structure}.
Here, it is important to remark that we do \textit{not} fix the topology of spacetime a priori, but allow variations over all subset topologies of $\F$.

For further considerations, let $x \in M$ be a spacetime point of \ac{cfs}.
It is convenient to introduce the spin space $S_x$ and its orthogonal projection $\pi_x$ by
\begin{align*}
    S_{x} &\defeq x(\H) \subset \H && \text{subspace of dimension $\leq 2n$} \\
    \pi_{x} &: \H \rightarrow S_{x}\subset \H && \text{orthogonal projection onto $S_{x}$} \:.
\end{align*}
With that notation at hand, we introduce the physical wave functions
\begin{align*}
    \psi^u(x) = \pi_{x} u \quad \text{ for } \quad u \in \H \:.
\end{align*}
Here $u$ is a vector in the Hilbert space~$\H$.
As a result of this construction, for every point~$x \in M$,
the physical wave function $\psi^u(x)$ is a $2n$ dimensional vector that is an element in the vector space $S_{x}$.
One can think of the vector space $S_{x}$ as a fiber space attached to a point $x$ in the operator manifold.
Although physical wave functions are defined directly on the \ac{cfs} spacetime $M$, the local correlation map allows us
to give an explicit realization of the Hilbert space in the classical manifold $\mathcal{M}$ in terms of sections in a suitable fiber bundle.
In summary, the points in the set $M$ are called \textit{spacetime points}, and the entire set $M$ is referred to as \textit{spacetime}.
This notion of spacetime turns out to be quite different from standard Lorentzian spacetime.
Many familiar features, such as the one of a smooth structure, metric, etc.~are absent in the notion of spacetime used in \ac{cfs}.

As noted above, spacetime points $x, y, \dots \in M \subset \F \subset \Lin(\H)$ in \ac{cfs} are linear operators
on $\H$.
They encode the local correlation information and have additional structure compared to classical spacetime points, namely
their eigenvalues and eigenspaces.
We now proceed to summarize those internal structures presented in~\cite{cfs}.
Later in the paper, we build on these structures and try to reformulate them in a language more familiar to those of other fields of physics.
In particular, we focus on the fact that starting from the abstract definition of a \ac{cfs}, one obtains the usual
spacetime structures such as causal relations.
To achieve this, given a \ac{cfs},
additional mathematical objects are introduced that are \textit{inherent} in the sense that they
only use information already encoded in the CFS\footnote{
    One can show that these inherent structures agree with familiar structures in a classical spacetime in suitable limiting cases.
}.
Our aim in the present paper is to expand on these inherent structures to establish a first proposal for a complete ontology of the theory that
does not rely on emergent structures of the continuum limit.

To define a causal structure in spacetime, we study the spectral properties of the operator product~$xy$.
This operator product is still an operator of rank at most~$2n$.
However, in general it is not symmetric (note that~$(xy)^* = yx \neq xy$ unless the
operators commute) and therefore not an element of $\F_n$.
We denote the rank of~$xy$ by~$k \leq 2n$.
Counting algebraic multiplicities, we choose~$\lambda^{xy}_1, \ldots, \lambda^{xy}_{k} \in \C$ as all the non-zero eigenvalues
and set~$\lambda^{xy}_{k+1}, \ldots, \lambda^{xy}_{2n}=0$.
These eigenvalues give rise to the following notion of a causal structure.
\begin{definition}
    \label{def:causalstructure}
    The points~$x,y \in \F_n$ are said to be
    \begin{align*}
        \left\{ \begin{array}{cll}
                    \text{{\textbf{spacelike}} separated} & & \text{if } \abs{\lambda^{xy}_j} = \abs{\lambda^{xy}_k} \text{for all } j, k = 1, \ldots, 2n \\
                    \text{{\textbf{timelike}} separated} & & \text{if } \lambda^{xy}_1, \ldots, \lambda^{xy}_{2n} \text{ are all real} \\
                    & & \text{and } \abs{\lambda^{xy}_j} \neq \abs{\lambda^{xy}_k} \text{ for some } j, k  \\[0.2em]
                    \text{{\textbf{lightlike}} separated} & & \text{otherwise}\:.
        \end{array}\right.
    \end{align*}
\end{definition} \noindent
More specifically, the points~$x$ and~$y$ are lightlike
separated if not all the eigenvalues have the same absolute value and if not all of them are real.
We point out that this notion of causality does not rely on an underlying metric.

We now come to the core of the theory of \ac{cfs}s: the \textit{causal action principle}.
In order to distinguish the physically admissible \ac{cfs}, one must formulate constraints in the form of physical equations.
To this end, we require that the measure $\rho$ be a minimizer of the causal action $\mathcal{S}$ defined by
\begin{align}
    \text{\textit{Lagrangian:}} && \L(x,y) &= \frac{1}{4n} \sum_{i,j=1}^{2n} \left( \abs{\lambda^{xy}_i} - \abs{\lambda^{xy}_j} \right)^2 \label{eq:lagrangian} \\
    \text{\textit{causal action:}} && S[\rho] &= \iint_{\F_n \times \F_n} \L(x,y) \dd\rho(x) \dd\rho(y)
\end{align}
under the following constraints
\begin{align*}
    \text{\textit{volume constraint:}} && \rho(\F_n) &= 1, \\
    \text{\textit{trace constraint:}} && \int_{\F_n} \tr(x) \dd\rho(x) &= 1, \\
    \text{\textit{boundedness constraint:}} && \int_{\F_n} \left( \sum_{j=1}^{2n} \abs{\lambda^{xy}_j} \right)^2 \dd\rho(x) &\leq \text{const}.
\end{align*}
This variational principle is mathematically well-posed if~$\H$ is finite-dimensional (for details, see~\cite[Section~1.1.1]{cfs}).

A minimizer satisfies the \textit{\ac{el} equations} which state that for suitable parameters $\kappa, \mathfrak{r}, \s > 0$
(Lagrange multipliers of the boundedness, trace and volume constraints),
the function $\ell: \F_n \rightarrow \R_0^+$ defined by
\begin{align}
    \label{eq:el-equation-l}
    \ell(x) \defeq \int_M \L(x,y) \dd\rho(y) + \kappa \int_M \left( \sum_{j=1}^{2n} \abs{\lambda^{xy}_j} \right)^2 \dd\rho(y) - \mathfrak{r} \tr x - \s
\end{align}
is minimal and vanishes in spacetime, i.e.
\begin{align}
    \label{eq:euler-largange}
    D\ell|_M = 0 && \ell|_M = 0
\end{align}
These equations describe the dynamics of the \ac{cfs}.
In a suitable limiting case, referred to as the \textit{continuum limit}, the \ac{el} equations
give rise to the equations of motion for the fields in the standard model and \ac{gr}, for details see~\cite{cfs}.

We finally give more details on the notion of causality.
In Definition~\ref{def:causalstructure} we already introduce notions of timelike, spacelike and lightlike separation.
The Lagrangian~\eqref{eq:lagrangian} is compatible with this notion of causality in the following sense.
Suppose that two points $x, y \in \F_n$ are spacelike separated.
Then the eigenvalues $\lambda^{xy}_i$ all have the same absolute value, implying that the Lagrangian~\eqref{eq:lagrangian} vanishes.
As a consequence, points with spacelike separation do not contribute to the action and therefore drop out of the
\ac{el} equation~\eqref{eq:euler-largange}.
This is in correspondence with the usual notion of causality, where points with spacelike separation cannot influence each other.

We finally point out that the causal structure is not necessarily transitive.
This corresponds to our physical conception that the transitivity of the causal relations could be violated both on the
cosmological scale (there might be closed timelike curves) and on the microscopic scale (there seems no compelling reason
why the causal relations should be transitive down to the Planck scale); see, e.g.~\cite{anselmi2007renormalization,anselmi2019fakeons,anselmi2018quantum} for microscopic causality violations
in the standard QFT context.
However, in~\cite{Dappiaggi2020, Finster2023} causal cone structures were constructed, which do give rise to transitive causal relations
(for details, see, in particular,~\cite[section 4.1]{Dappiaggi2020}).
All these causal relations coincide on macroscopic scales much larger than the Planck length.

\section{A Viable Ontology of CFS} \label{sec:ontology-of-cfs}

In the Definition~\ref{def:cfs} of a \ac{cfs} the Hilbert space $\H$ consists only of occupied states of the physical system.
This leads to an unusual situation where the entire Hilbert space is ``occupied''.
Conventionally, one thinks of the state of a physical system to be represented by a single vector in, or a density matrix
on a Hilbert space.
To better align ~\ac{cfs} with conventional thinking, we introduce the following definition of an auxiliary Hilbert space.

\begin{definition}[Auxillary Hilbert space]
    Let $(\H, \F_n, \rho)$ be a \ac{cfs} and $\tilde{\H}$ a Hilbert space with $\H \subset \tilde{\H}$, $\dim {\tilde{\H}} > \dim {\H}$.
    Then $\tilde{\H}$ is called an \textbf{auxiliary Hilbert space} of $\H$.
\end{definition}

Although the addition of an auxiliary Hilbert space might seem arbitrary, it serves an important technical purpose.
Suppose we want to do perturbation theory in \ac{cfs}.
In this case, we start with two \ac{cfs} with Hilbert spaces $\H_1$ and Hilbert spaces $\H_2$.
To be able to formulate one as the perturbation of the other, we need to be able to relate the two Hilbert spaces.
The way to do this is by embedding both of them into one common larger, i.e.~auxiliary Hilbert space.
It has been shown in~\cite{Finster2025} that there exists a canonical construction for this.
Furthermore, in Section~\ref{sec:continuum-limit} we see that in the example of the Minkowski spacetime, this construction
complements the Dirac sea with the positive energy solutions of the Dirac equation.

Given an auxiliary Hilbert space $\tilde{\H}$ we can return to the fundamental Hilbert space of the \ac{cfs} through a projection  $\Omega: \tilde{\H} \to \H$ for which
\begin{align*}
    \tr_{\tilde{\H}} \left[ \Omega \right] = \tr_\H \left[ \id_\H \right] = \dim \H = N \,
\end{align*}
holds.
Since in \ac{cfs} the Hilbert space $\H$ contains only the occupied states, the dimension $N$ of $\H$ can be interpreted as the number of particles constituting the \ac{cfs}.
Therefore, we say that the $N$ particle state $\Omega$ describes the total physical system in the auxiliary Hilbert space.
Hence, $(\H, \Omega)$ can be thought of as a many-body quantum system represented by the maximally mixed
density matrix $\frac{\Omega}{N}$ on $\tilde{\H}$.

These concepts can be generalized to arbitrary projections onto subspaces of $\H$.
\begin{definition}[(Sub)System]
    Let $(\H, \F_n, \rho)$ be a~\ac{cfs}, $\tilde{\H}$ an auxiliary Hilbert space and $\H_A \subseteq \H$ a sub-Hilbert space of $\H$.
    Then a \textbf{subsystem} is the tuple $(\H_A, \omega_A)$, where $\omega_A: \tilde{\H} \to \H_A$ is the projection onto $\H_A$.
    In addition, the \textbf{particle number} of a subsystem $(\H_A, \omega_A)$ is defined as $\tr_{\tilde{\H}} \omega_A$.
    We call $(\H, \Omega)$ the \textbf{total system}.
\end{definition}

By this definition, each state (normed vector) $\ket{u} \in \H$ corresponds to a one-particle subsystem $(\H_u, \omega_u) = (\Span\left\{ \ket{u} \right\}, \ket{u}\bra{u})$.
In particular, since $\H$ is a Hilbert space, it admits an orthonormal basis $\left( \ket{u_i} \right)$.
Thus, $\Omega$ can be expressed as
\begin{align}
    \Omega = \sum_{i = 1}^N \ket{u_i}\bra{u_i}.
\end{align}
For notational clarity, we denote the one-particle subsystem corresponding to the $i$-th basis vector by $(\H_i, \omega_i)$.

This reformulation as a quantum many-body system, which is used for the description of condensed matter physics, connects
the \ac{cfs} picture of the Dirac sea with the more familiar picture of solid-state theory.

In a crystalline material, you have a band structure of admissible states in the so-called Brillouin zone, which is the
Fourier transform of the lattice cell.
These bands are then filled with electrons.
At zero temperature, the filled states are just the $N$ lowest energy states.
The state with the highest energy is then referred to as the Fermi energy of the material; see, e.g.~\cite{solidstate,sigrist2013solid,wen2004quantum} for an introduction and further references.
The material properties, e.g.~when heated of under an external potential, are then largely determined by the properties
of the band structure near the Fermi energy.
In the case of Minkowski space as a \ac{cfs} which we discuss below, the mass-shells of the different fermionic generations
play the role of the band structure in a ``Brillouin zone'' corresponding to a cell with infinite size\footnote{This suggests that \ac{cfs} is firmly in the ``emergent gravity'' sector of approaches to unification~\cite{crowther2019renormalizability}. However, it is not entirely clear to us, how the structures of \ac{cfs} compare to those a set of ideas referred to as ``condensed matter'' approach~\cite{bain2014three}, first put forward by Sakharov~\cite{sakharov2000vacuum}, the discussion of which has been translated into modern laguage by Visser in \cite{visser2002sakharov}. As pointed out by~\cite{bain2014three}, ``according to Wen \cite[p.341]{wen2004quantum}, topological orders characterize 
any condensate with ground states that possess a finite energy gap''. This is interesting in the context of \ac{cfs}, given the fact that in the continuum limit described by the Dirac sea, the minimizer of the causal action principle actually possesses a finite energy gap, namely the mass-gap, and the Dirac sea can be thought of as a fermionic condensate.}.

One peculiarity of the solid-state picture that shows up here as well is the fact that one is effectively working in a
first quantized picture.
That means that we do not have the antisymmetric structure available, which in the case of the fermionic Fock space of QFT
guarantees that the Pauli exclusion principle is obeyed.
In solid-state theory, this is instead put in more or less by hand when the states in the band structure are filled by
assuming a Fermi-Dirac distribution of the occupied states.
This is similar to the way the electron shells are ``filled'' in atomic physics.

In \ac{cfs} we employ a similar first quantized picture in the derivation of the classical continuum limit.
However, as we show now, thanks to the way the theory is built, the Pauli exclusion principle is, in fact, baked into the
elementary definition of the theory.
Looking at the \ac{cfs} as an $N$ particle quantum system $(\H, \Omega)$ with $ \tr_{\tilde{\H}} \left[ \Omega \right] = N$ we can
ask ourselves what the occupation number of any state $\ket{u} \in \tilde{\H}$ is, i.e.~what the probability is to find a
particle in that state.

\begin{definition}[Occupation number]
    Let $\ket{u} \in \tilde{\H}$ be a state and $(\H_A, \omega_A)$ be a (sub)system, then the \textbf{occupation number} of $\ket{u}$ in $(\H_A, \omega_A)$ is defined as $\braket{u}{\omega_A u}$.
\end{definition}

Based on this definition, the following proposition shows immediately that, independent of the choice of an auxiliary Hilbert space, all \ac{cfs}
satisfy the Pauli exclusion principle.

\begin{proposition}[Pauli Exclusion Principle]
    Let $\ket{u} \in \tilde{\H}$ be a state and $(\H_A, \omega_A)$ a (sub)system, then the occupation number satisfies
    \begin{align}
        0 \leq \braket{u}{\omega_A u} \leq 1.
    \end{align}
\end{proposition}
\begin{proof}
    Since $\omega_A$ is a projection, the operator norm satisfies
    \begin{align}
        \opnorm{\omega_A} \defeq& \sup \left\{ \braket{u}{\omega_A u} : \ket{u} \in \tilde{\H}, \braket{u}{u} = 1 \right\} \nonumber \\
        =& \sup \left\{ \braket{u}{u} : \ket{u} \in \H_A, \braket{u}{u} = 1 \right\} = 1.
    \end{align}
    Thus, $\braket{u}{\omega_A u} \leq \opnorm{\omega_a} \braket{u}{u} = 1$.
    The non-negativity of the occupation number follows directly from the fact that the projection $\omega_A$ is a positive
    operator.
\end{proof}

For any state $\ket{u} \in \tilde{\H} \setminus \H$ ``outside'' the total system $(\H, \Omega)$,
the occupation number is obviously zero for all (sub)systems $(\H_A, \omega_A)$.
Only states, which contain contributions from $\H$, have non-zero occupation number in $(\H, \Omega)$.
More precisely, every state $\ket{u} \in \H$ has occupation number $1$ in $(\H, \Omega)$, i.e.~all states in $\H$ are occupied.
In particular, this statement holds even though $\Omega$ is not a density matrix, but an $N$-particle state with $N = \tr \Omega$.

The fact that the construction of a \ac{cfs} guarantees that the Pauli exclusion principle is satisfied when we interpret the
basis vectors in $\H$ as occupied particle states in $\tilde{\H}$ justifies thinking of a \ac{cfs} as a quantum system
made up of $N$ fermions\footnote{
    Under the assumption that every state in $\H$ is represented by a physical wavefunction and every state in $\tilde{\H} \setminus \H$ cannot be represented by a physical wavefunction, our definition of occupation number agrees with that in~\cite[Section 5.9]{Finster2024}, for vectors in $\H$ and $\tilde{\H} \setminus \H$ respectively.
    Since our setup also allows for (normed) linear combinations of states in $\H$ and $\tilde{\H} \setminus \H$ our occupation number takes values in $[0,1]$.
    Thus, our definition aligns better with the idea of an underlying Fermi-Dirac distribution. 
    In the context of our discussion in Section~\ref{subsec:reference frames} this might allow for a basis-dependent notion of a temperature for a \ac{cfs} akin to the Unruh effect~\cite{Unruh}. The analysis in~\cite{finster2017fermionic} could serve as a starting point for this investigation. 
}. Furthermore the fact that the above proof guarantees, that any \ac{cfs} respects the Pauli exclusion principle, means that the theory is compatible with any experimental bounds obtained, e.g., by the VIP collaboration~\cite{baudis2024search,marton2015high,piscicchia2022strongest,piscicchia2023experimental} on hypothetical violations thereof. 

In contrast, this restriction of the occupation number does not apply to the emergent fields, i.e.~all the (bosonic) interaction fields as they emerge as a
linearized perturbation of the fermionic $N$-particle \ac{cfs}.
Accordingly, the occupation number of such perturbations is not constrained from within the theory itself.
This explains why particles corresponding to an element in $\H$ are subject to constraints that do not apply to emergent
fields, such as photons\footnote{
    At present it is unclear whether \ac{cfs} provides a fundamental reason why the physical wave functions of elements in $\H$
    have to satisfy a Dirac equation in the effective description of the continuum limit, it just turns out to be the case in
    the known examples.
    Hence, for the moment these are just two coinciding facts looking for a joint explanation.
    One possible approach is to try and apply the spin-statistics theorem ``in reverse'' to conclude from the statistics of the physical wave functions in the continuum limit in Minkowski space that they have to correspond to spin $\frac{1}{2}$ particles.
}.

\subsection{Observables}\label{subsec:observables}

Throughout this section, we work with the physical Hilbert space $\H$.
One could equally well work with the auxiliary Hilbert space $\tilde{\H}$ without affecting the results as should be
for an auxiliary structure.

The key question for any physical theory is what its observables are.
As mentioned in Section~\ref{sec:preliminaries} the goal of the present paper is to work entirely with structures inherent
to \ac{cfs}.
In~\cite{dzhunushaliev2023quantum, dzhunushaliev2023quantization} operator-valued measures were considered.
These come naturally with the inherent structures of \ac{cfs} which leads to the following definition\footnote{
    There are more general observables one could define, using the spacetime point operators and its eigenspaces,
    e.g.~$\tilde{O}(\mathcal{U}) =  \int_\mathcal{U} x\;  \dd \rho(x)$.
    However, for our purposes here we currently believe that the fermionic position observables are sufficient for a complete ontology.
    In Section~\ref{sec:causalcorrelations} we additionally introduce the causal correlation operators.
    However, it is unclear whether they are formally observables as in general they are not self-adjoint operators despite having real eigenvalues, furthermore they lack some of the important properties of these position observables.
}.

\begin{definition}[Position Observables] \label{def:observables}
    Let $(\H, \F_n, \rho)$ be a \ac{cfs} and $\mathcal{U} \subseteq \F_n$, then the \textbf{observable} of the region $\mathcal{U}$ is the operator
    \begin{align}
        \mathcal{O}(\mathcal{U}) \defeq \int_\mathcal{U} \pi_x \dd \rho(x).
    \end{align}
    The \textbf{expectation value} of an operator $\mathcal{O}: \H \to \H$ for a (sub-) system $(A, \omega_A)$ is defined as
    \begin{align}
        \expval{\mathcal{O}}_{\omega_A} \defeq \frac{1}{2n} \tr_{\H} \left[ \omega_A \mathcal{O} \right].
    \end{align}
\end{definition}
For a state $\ket{u} \in \H$, the expectation value of an operator $\mathcal{O}$ conincides up to a factor of $\frac{1}{2n}$
with the standard notion:
\begin{align}
    \expval{\mathcal{O}}_{u} \defeq \expval{\mathcal{O}}_{\omega_u} = \frac{1}{2n} \tr_\H (\omega_u \mathcal{O}) = \frac{\braket{u}{\mathcal{O} u}}{2n}.
\end{align}

In addition, these definitions of observables and expectation values have the following properties.
\begin{proposition}
    \label{prop:positivity-of-the-spacetime-observables}
    Let $\mathcal{U} \subseteq \F_n$, then the observable of $\mathcal{U}$ satisfies
    \begin{enumerate}
        \item $\displaystyle \expval{\mathcal{O}(\mathcal{U})}_{\omega_A} \geq 0$ for every subsystem $(\H_A, \omega_A)$,
        \item $\displaystyle \mathcal{O}(\mathcal{U}) = 0 \Leftrightarrow M \cap \mathcal{U} = 0$.
    \end{enumerate}
\end{proposition}
\begin{proof}
    ``(1)'': Let $\ket{u} \in \H$ be a state, then
    \begin{align*}
        \expval{O(\mathcal{U})}_{u} &= \frac{1}{2n} \braket{u}{O(\mathcal{U}) u} \\
        &= \frac{1}{2n} \int_\mathcal{U} \braket{u}{\pi_x u} \dd \rho(x) = \frac{1}{2n} \int_\mathcal{U} \braket{\pi_x u} \dd \rho(x)\\
        &= \frac{1}{2n} \int_\mathcal{U} \braket{\psi^{u}(x)} \dd \rho(x) = \frac{1}{2n} \int_\mathcal{U} \norm{\psi^{u}(x)}^2 \dd \rho(x) \geq 0
    \end{align*}
    Since any subssystem can be expressed as a sum of one particle systems, we have $\expval{\mathcal{O}(\mathcal{U})}_{\omega_A} \geq 0$.

    \noindent
    ``(2)'': Since $\pi_x \neq 0$, we have that $\mathcal{O}(\mathcal{U}) = 0$ implies $\mathcal{U}$ is a $\rho$-zero set,
    i.e.~$\rho(\mathcal{U}) = 0$.
    Thus, $\mathcal{U}$ and $M = \supp \rho$ are disjoint.
    On the other hand, if $\mathcal{U} \cap M = \emptyset$, we have
    \begin{align*}
        \mathcal{O}(\mathcal{U}) = \int_{\mathcal{U}} \pi_x \dd \rho(x) = \int_{\mathcal{U} \cap M} \pi_x \dd \rho(x) = 0.
    \end{align*}
\end{proof}

Property (2) is an important consistency check that we are indeed dealing with a sensible definition of physical observables.
It guarantees that all position observables are trivial outside of spacetime.
Hence, the expectation value of observing any particles outside of spacetime is zero, leading to the following corollary.
\begin{corollary}
    Let $\mathcal{U} \subset {\F_n}$ be an arbitrary subset then we have that for every subsystem $(\H_A, \omega_A)$
    \begin{align}
        \expval{\mathcal{O}(\mathcal{U})}_{\omega_A} = \expval{\mathcal{O}(\mathcal{U} \cap M)}_{\omega_A}
    \end{align}
\end{corollary}

This corollary allows us in the following considerations to choose $\mathcal{U} \subset {\F_n}$ without regard to whether
$\mathcal{U}$ is a proper subset of spacetime $M$.
Additionally, we find that the expectation value of the spacetime region observable $\mathcal{O}(\mathcal{U})$ is a non-negative
real number.

From a physical point of view, the integral $\int_\mathcal{U} \norm{\psi^{u}(x)}^2 \dd \rho(x)$ as it appears in the proof
of proposition~\ref{prop:positivity-of-the-spacetime-observables} reminds us of Born's rule, where the integrand
$\norm{\psi(x)}^2$ is interpreted as the probability density of finding a particle with wave-function $\psi(x)$ at a spacetime point $x$.
However, in contrast to Born's rule where one only considers spatial hyper-surfaces, $\mathcal{U}$ is a general spacetime region.
In the continuum limit, a spatial hypersurface has measure zero.
However, as shown in~\cite{finster2013non}, for solutions of the Dirac equation there is a relation between a spacetime inner product and a hypersurface scalar
product for spinor fields.
This raises the following conjecture in the continuum limit.

\begin{conjecture}[Born's Rule]
    Let $(\mathcal{M}, g)$ be a classical globally hyperbolic spacetime.
    Let $B$ be a compact simply connected subset of a Cauchy surface $\mathcal{N}$.
    Suppose that $\mathcal{U} = D(B)$ is the domain of dependence of $B$.
    Then~\cite[Proposition 3.1. ]{finster2013non} can be applied to pull back the spacetime observable of $\mathcal{U}$ to an observable on $B$.
\end{conjecture}

This conjecture would allow for a direct connection with Born's rule.
This conjecture would also shine light on the problem of localization in quantum mechanics; see e.g.~\cite{hegerfeldt1974remark,skagerstam,finster2023incompatibility, castrigiano2017dirac,perez1977localization, beck2021local}.
Only a very distinct combination of spacetime observable and choice of hypersurface allows for an interpretation of an
observable localized in space in $B$.
For generic $\mathcal{U}$ and generic Cauchy surfaces, this identification fails.
\begin{proposition}
    For a measurable subset of regular correlation operators $\mathcal{U} \subset{\F_n}^{\text{reg}}$, the expectation value of
    the position observable $\mathcal{O}(\mathcal{U})$ for the total system $(\H, \Omega)$ is the volume of the spacetime region $\mathcal{U} \cap M$
    \begin{align}
        \expval{\mathcal{O}(\mathcal{U})}_{\Omega} = \rho(\mathcal{U}).
    \end{align}
\end{proposition}
\begin{proof}
    The proof is a straightforward computation
    \begin{align*}
        \expval{\mathcal{O}(\mathcal{U})}_{\Omega} = \frac{1}{2n} \int_{\mathcal{U}} \tr \pi_x \dd \rho(x) = \rho(\mathcal{U})
    \end{align*}
    using the fact that summing over a basis of $\H$ at every point precisely gives us the trace of the projection operator.
\end{proof}
This result suggests thinking of the states as carriers of spacetime volume.
The higher the expectation value to observe a particle in a certain spacetime region $\mathcal{U}$, the bigger the total
volume of the said spacetime region.

\subsection{Spacetime Superposition}\label{subsec:spacetime-superposition}

This observation above leads us to introduce the idea of the one-particle measure and its associated one-particle spacetime.

\begin{definition}\label{def:subsystem}
    Let $(\H_A, \omega_A)$ be a (sub-) system, then the measure assigned to this (sub-) system is defined by
    \begin{align}
        \rho_{\omega_A} (\mathcal{U}) \defeq \expval{\mathcal{O}(\mathcal{U})}_{\omega_A} \text{ for measureable } \mathcal{U} \subseteq \F
    \end{align}
    and the corresponding spacetime is defined by
    \begin{align}
        M_{\omega_A} \defeq \supp \rho_{\omega_A}.
    \end{align}
\end{definition}

Similarly to the notation of the one-particle subsystems, we denote the \textbf{one-particle measure} and \textbf{one-particle spacetime}
associated with a state $\ket{u} \in \H$ by $\rho_{u} \defeq \rho_{\omega_u}$ and $M_u \defeq M_{\omega_u}$ respectively.
For a basis $\left( \ket{u_i} \right)$ of $\H$, we write $\rho_i \defeq \rho_{u_i}$ and $M_i \defeq M_{u_i}$ accordingly.
The one-particle measure and spacetime have the following property:

\begin{proposition}\label{prop:superposition}
    Let $\left( \ket{u_i} \right)$ be a basis of $\H$, then the following statements hold
    \begin{enumerate}
        \item $\displaystyle \sum_{i = 1}^{N} \rho_i(\mathcal{U}) = \rho(\mathcal{U})$ for every measurable subset $\mathcal{U} \subseteq \F^{\text{reg}}$,
        \item $\displaystyle \bigcup_{i = 1}^N M_i = M$.
    \end{enumerate}
\end{proposition}

\begin{proof}
    First, we have that
    \begin{align*}
        \rho_i(\mathcal{U}) = \expval{\mathcal{O}(\mathcal{U})}_{u_i} = \frac{1}{2n} \int_{\mathcal{U}} \tr_\H \left[ \ket{u_i} \bra{u_i} \pi_x \right] \dd \rho(x) = \frac{1}{2n} \int_{\mathcal{U}} \bra{u_i} \pi_x \ket{u_i} \dd \rho(x).
    \end{align*}
    Summing over all basis vectors then gives
    \begin{align*}
        \sum_{i = 1}^{N} \rho_i(\mathcal{U}) = \frac{1}{2n} \int_{\mathcal{U}} \tr \pi_x \dd \rho(x) = \rho(\mathcal{U}).
    \end{align*}

    Secondly, as every one-particle measure satisfies $\rho_i \geq 0$, we have $\rho(\mathcal{U}) = 0 \Leftrightarrow \rho_i(\mathcal{U}) = 0$ for all $i$.
    Thus, $\supp \rho = \bigcup_{i = 1}^N \supp \rho_i$.
\end{proof}

Therefore, we can consider the $N$ particle spacetime $M$ as a superposition of the $N$ one-particle spacetimes $M_i$
corresponding to a basis $\left( \ket{u_i} \right)$ of $\H$. We would like to point out that this decomposition corresponds to a subclass of what is referred to as   ``fragmentation'' of the measure in~\cite{finster2019positive,dappiaggi2024holographic} in the context of the derivation of the quantum field theory limit of \ac{cfs}. It would be interesting to understand whether our new perspective here allows a better understanding of the holographic mixing associated with the fragmentation of the measure.

It is worth noting here that the different one-particle spacetimes $M_u$ might have very different topologies compared to $M$,
as they are different subsets of $\F$ from which they inherit the subset topology.
The question of which spacetime topologies should be considered for the variational principle in quantum theories of gravity has been
subject to debate~\cite{carlip1993sum,loll2006nonperturbative,glaser2019quantum,buoninfante2024visions}.
Proposition~\ref{prop:superposition} shows that for \ac{cfs} spacetime topologies to be ``summed over'' are determined by
the minimizing measure and, thereby, an \textit{output} of the action principle\footnote{
This holds true in a weaker sense for theories that treat the measure as an independent variable as this admits variation over all subset topologies of the manifold under consideration.
}.

Furthermore, in light of various arguments in the literature, e.g.~\cite{jiatimeorder, hartle1995spacetime,lightconefluctuations,ford1999spacetime,hardy2005probability,hardy2007towards,zych2019bell,zychquantumtemporal,zychblackhole,foo2021schrodinger,foo2023quantum,foo2023relativity}
regarding the superposition of causal structures in quantum gravity, it is worth noting the following corollary.

\begin{corollary}\label{cor:fixcausalstructure}
    Let $M_u$ and $M_v$ be two one-particle spacetimes with $M_u \cap M_v \neq \emptyset$.
    Then the nature of the causal relation between two spacetime points $x, y \subset M_u \cap M_v$ is independent of whether we consider them as points in $M_u$ or as points in $M_v$.
\end{corollary}
\begin{proof}
    This simply follows from the fact that the causal relation between $x$ and $y$ is defined on the underlying operator manifold $\Freg$.
    The (one-particle) spacetimes then simply inherit this underlying causal relation.
\end{proof}

The last concept we want to introduce in this section is the notion of localized and delocalized states in \ac{cfs}.
This distinction plays an important role in the probe-background split that we introduce in the conceptual description
of experiments in Section~\ref{sec:experiment}.

\begin{definition}[Localization of States]\label{def:localization-of-states}
    Let $\ket{u} \in H$ and $\mathcal{U} \subseteq M$, then $\ket{u}$ is \textbf{localized in} $\mathcal{U}$, if
    for every $x \in M_u$ there exists a $y \in \mathcal{U}$ such that $x$ and $y$ are causally separated.
\end{definition}

This definition extends the concept of Cauchy surfaces.
Suppose that we have a globally hyperbolic spacetime with global time function $t$, then every wavefunction
is supported for all $t$.
Thus, it intersects every Cauchy surface.
Definition~\ref{def:localization-of-states} also allows for (compact) spacetime regions instead of Cauchy surfaces (see Figure~\ref{fig:localized-state}).

\begin{figure}
    \label{fig:localized-state}
\begin{tikzpicture}
    \draw[->] (0,-0.5) -- (0,4) node[above]{$t$};
    \draw[->] (-0.5,0) -- (12,0) node[right]{$\vec{x}$};

    \filldraw[fill=gray!30, draw=black]
    (5,0) -- (1, 4) -- (10, 4) -- (6, 0) -- (6.5, -0.5) -- (4.5, -0.5) -- cycle;
    \node at (5.5,3) {$M_u = \supp \Psi^u$};

    \filldraw[fill=gray!70, draw=black]
    plot[smooth cycle, tension=1] coordinates {
        (1, 1)
        (2, 2)
        (4, 1)
        (6, 1.5)
        (8, 1)
        (3, 0.5)
    };
    \node at (2, 1) {$\mathcal{U}$};
\end{tikzpicture}
    
    \caption{
        This visualization shows associated one particle spacetime $M_u$ of a state $\ket{u} \in H$.
        The spacetime region $\mathcal{U}$ intersects $M_u$ completely.
        Thus, by Definition~\ref{def:localization-of-states}, $\ket{u}$ is localized in $\mathcal{U}$.
        In addition, $\ket{u}$ is not delocalized in the sense of Definition~\ref{def:delocalized-states}. To help with intuition we chose the full spacetime to be the continuum limit of Minkowski space. In general $\mathcal{U}$ and $M_u$ are subsets of an abstract \ac{cfs} spacetime $M$.
    }
\end{figure}
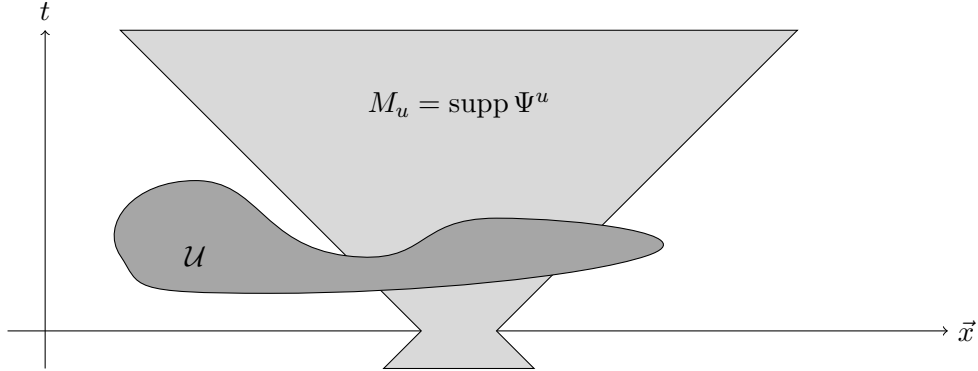

Colloquially, we say that we can observe a particle $\psi^{v}(x)$ in a region $\mathcal{U}$ if $\ket{u}$ is localized in this region.
We define delocalized states in contrast as follows.
\begin{definition}[Delocalized states]\label{def:delocalized-states}
    Let $\ket{u} \in \H$, then $\ket{u}$ is \textbf{delocalized} if $M_u = M$.
\end{definition}
This implies that for a delocalized state the observables for any spacetime region is non-zero.
For such states, there exists no meaningful way to speak about ``where'' they are.
As we argue in Section~\ref{sec:experiment}, these states cannot be used as probes to conduct experiments.

\subsection{Causal Correlations}\label{sec:causalcorrelations}

In this section, we consider a basis $\left( \ket{u_i} \right)$ of $\H$.
Then, we can study the amount of information a physical wave function $\psi^{u_i}(x)$ contains about the correlation between
two points $x$ and $y$.
The operator
\begin{align}\label{eq:informationtransferoperator}
    \spinket{\psi^{u_i}(x)} \spinbra{\psi^{u_i}(y)}
\end{align}
defines how much information about an observable at $y$ can be extracted via the physical wave function $\psi^{u_i}(x)$ of the state $\ket{u_i}$.
More precisely, we want to answer the following question: suppose we have an observable $O_y: S_y \to S_y$ at a point $y$,
what information can an observer at a point
$x$ learn about $O_y$ if $\psi^{u_i}(x)$ is the only information channel available to them?
Using \eqref{eq:informationtransferoperator} as information transfer operator related to the state $\ket{u_i}$, the corresponding
information retained at $x$ is then given by
\begin{align}
    \tilde O_{xy}(\psi^{u_i}) = \spinket{\psi^{u_i}(x)} \spinbra{\psi^{u_i}(y)}O_y \spinket{\psi^{u_i}(y)} \spinbra{\psi^{u_i}(x)} : S_x \rightarrow S_x.
\end{align}
To better motivate this notion, we revert back to non-relativistic quantum mechanics for a moment.
If we carry out the same considerations for an operator $O_t$ at time $t$ and a quantum state $\psi(t)$, we get
\begin{align}
    \label{eq:heisenberg-with-states}
    \tilde O_{t't}=   \ket{\psi(t')} \bra{\psi(t)}O_t \ket{\psi(t)} \bra{\psi(t')}.
\end{align}
If the evolution of the state is unitary, then we get that
\begin{align}
    \ket{\psi(t')} \bra{\psi(t)}= U_{(t-t')} \text{ and } \ket{\psi(t)} \bra{\psi(t')}= U_{(t'-t)}
\end{align}
and with
\begin{align}
    U_{-t} = U^{-1}_t.
\end{align}
equation~\eqref{eq:heisenberg-with-states} becomes simply Heisenberg evolution for observables
\begin{align}
    \label{eq:heisenberg-with-operators}
    \tilde O_{t'}= U_{(t'-t)} ^{-1}  O_t U_{(t'-t)} .
\end{align}
and unitarity guarantees preservation of the entire information between the observable $O_t$ and $O_{t'}$ for any state $\psi(t)$ serving as a communication channel between time $t$ and time $t'$.
However, for two arbitrary points in spacetime, there is apriori no reason why one should expect preservation of information,
i.e.~for $\spinket{\psi^{u_i}(y)} \spinbra{\psi^{u_i}(x)}$ to be unitary.
We now remember Definition~\ref{def:observables} where we declared that, modulo the measure, from within \ac{cfs} the
only admissible observable at $y$ is
$O_y = \pi_y |_{S_y} = \id_{S_y}$.
Then, the information about this observable recovered at $x$ via the one particle wave function $\psi^{u_i}(x)$ is given by
\begin{align}
    \tilde O_{xy}(\psi^{u_i})=   \spinket{\psi^{u_i}(x)} \spinbraket{\psi^{u_i}(y)}{\psi^{u_i}(y)} \spinbra{\psi^{u_i}(x)} : S_x \rightarrow S_x.
\end{align}
To recover the total information shared between two spacetime points $x$ and $y$ in an $N$-particle \ac{cfs}, we again sum
over all basis vectors to get
\begin{align}
    A_{xy} \defeq \sum_{i, j = 1}^N \spinket{\psi^{u_j}(x)} \spinbraket{\psi^{u_j}(y)}{\psi^{u_i}(y)} \spinbra{\psi^{u_i}(x)}.
\end{align}
which turns out to be a familiar concept in \ac{cfs}, namely the so-called closed chain. Furthermore we can ``split'' the closed chain by introducing the kernel of the fermionic projector as the total information transfer operator between $x$ and $y$
\begin{align}
    \label{eq:kernel-of-the-fermionic-projector}
    P(x,y) \defeq \sum_{i = 1}^N \spinket{\psi^{u_i}(x)} \spinbra{\psi^{u_i}(y)}.
\end{align}
Reformulating the kernel of the fermionic projector and the closed chain in a basis independent way gives
\begin{align}
    P(x, y) &= \pi_x y|_{S_y} : S_y \to S_x, \\
    A_{xy} &= P(x, y) P(y, x) = \pi_x y x|_{S_x}  : S_x \to S_x.
\end{align}
Based on this representation, we find the following properties.

\begin{proposition} \label{prop:properties-of-closed-chain}
    Let $x, y \in M$ be two spacetime points, then
    \begin{enumerate}
        \item $P(x, y)$ is symmetric in the sense that $\left( P(x, y) \right)^* = P(y, x)$
        \item $A_{xy}$ is a symmetric operator on $S_x$.
    \end{enumerate}
    with respect to the spin inner product $\spinbraket{\cdot}{\cdot}_x$ and $\spinbraket{\cdot}{\cdot}_y$ respectively.
\end{proposition}
\begin{proof}
    ``(1)'': Let $\ket{u} \in S_x$ and $\ket{v} \in S_y$, then
    \begin{align*}
        \spinbraket{u}{P(x, y) v}_x &= \braket{u}{x P(x, y) v} = \braket{u}{x \pi_x y v}\\
        &= \braket{x u}{y v} = \braket{\pi_y x u}{y v} = \spinbraket{P(y, x) u}{v}_y.
    \end{align*}
    Thus, $\left( P(x, y) \right)^* = P(y, x)$.

    \noindent
    ``(2)'': The symmetry property of $A_{xy}$ follows directly from property (i) using that $A_{xy} = P(x, y) P(y, x)$.
\end{proof}

Moreover, the fact that the operator products $AB$ and $BA$ have the same non-vanishing eigenvalues (including multiplicities)
implies that the closed chain $A_{xy}$ and the operator product $xy$ have the same eigenvalues.
Thus, we can treat the Lagrangian as a function of the closed chain.
The information used in the causal action principle is then simply the information contained in the causal correlations
between all pairings $x$ and $y$ of the spacetime points.

As stated above, spacelike separated points (according to Definition~\ref{def:causalstructure}) do not
contribute to the Lagrangian in the action principle.
Hence, we can restrict the closed chain to causally separated points setting spacelike-separated points to zero.

\begin{definition}[Causal Correlation Operator]
    \label{def:causal-correlation-operator}
    Let $x, y \in M$, $A_{xy}$ be the closed chain between $x$ and $y$ as defined above, then the \textbf{causal correlation operator}
    is defined as
    \begin{align}
        \tilde{A}_{xy}: \H \to \H = \begin{cases}
                                          0 & \text{if $x$ and $y$ spacelike separated}\\
                                          A_{xy} \pi_x & \text{otherwise.}
        \end{cases}
    \end{align}
\end{definition}

Concerning the causal action, the causal correlation operator encodes the same information as the closed chain.
Thus, we do not lose any relevant physical information by using the causal correlation operator.

\begin{definition}[One-Particle Two-Point Correlation Strength]\label{def:1Pcorrelation}
    Let $\ket{u} \in \H$ and $x, y \in M$, then the \textbf{one-particle two-point correlation strength}
    $b_{u}(x, y)$ is defined as
    \begin{align}
        b_{u}(x, y) \defeq \frac{1}{2n} \tr_{\H} \left[ \ket{u} \bra{u} \tilde{A}_{xy} \right].
    \end{align}
\end{definition}

This definition is in agreement with the Definition~\ref{def:observables} of the expectation value of an operator, except for the fact
that $\tilde{A}_{xy}$ is not an observable in the sense of Definition~\ref{def:observables}
As a consequence, $b_u(x, y)$ is not necessarily a positive real number.
In general, $b_u(x, y)$ is a complex number with the following properties.
\begin{lemma}
    Let $x, y \in M$, then one-particle two-point correlation strength $b_{u}(x, y)$ satisfies
    \begin{enumerate}
        \item $\displaystyle \complexconj{b_u(x, y)} = b_u(y, x)$,
        \item if $A_{xy}$ is diagonalizable, $x$ and $y$ are timelike separated, then ${b_{u}(x, y) \in \R}$.
    \end{enumerate}
\end{lemma}
\begin{proof}
    ``(1)'':
    $b_u(x, y)$ can be rewritten as
    \begin{align*}
        b_u(x, y) &= \frac{1}{2n} \tr_{\H} \left[ \ket{u} \bra{u} \tilde{A}_{xy} \right] = \frac{1}{2n}  \braket{u}{\tilde{A}_{xy} u} = \frac{1}{2n} \braket{\psi^u(x)}{yx \psi^u(x)}.
    \end{align*}
    Thus, using that $x$ and $y$ are both symmetric operators, we get
    \begin{align*}
        \complexconj{b_u(x, y)} &= \complexconj{\braket{\psi^u(x)}{yx \psi^u(x)}} = \braket{yx \psi^u(x)}{\psi^u(x)} \\
        &= \braket{\psi^u(x)}{xy \psi^u(x)} = b_u(y, x)
    \end{align*}

    \noindent ``(2)'':
    If $A_{xy}$ is diagonalizable, then there also exists an orthonormal eigenbasis $\left( \ket{u_i} \right)$ of $\tilde{A}_{xy}$
    For timelike $x$ and $y$ all eigenvalues of $A_{xy}$ and hence of $\tilde{A}_{xy}$ are real (according to Definition~\ref{def:causalstructure}).
    This implies that in this case $b_u(x, y) \in \R$
\end{proof}

Similarly to the construction of the fermionic projector from the one-state information transfer operator and the analysis
of the expectation value of the observables as presented in Section~\ref{subsec:observables}, we obtain the total
two-point correlation strength by summing over all basis vectors $\ket{u_i}$.

\begin{definition}[Two-Point-Correlation-Strength]
    Let $\left( \ket{u_i} \right)_i$ be a basis of $\H$, $x, y \in M$, then the \textit{two-point-correlation-strength}
    is defined as
    \begin{align}
        b(x, y) \defeq \sum_{i = 1}^N b_{u_i}(x, y) = \expval{\tilde{A}_{xy}}_\Omega = \frac{1}{2n} \tr_\H {\tilde{A}_{xy}}.
    \end{align}
\end{definition}

In contrast to the one-particle two-point correlation strength, the total two-point correlation strength is a real number for all
$x$ and $y$.

\begin{lemma}
    The total two-point correlation strength $b(x, y)$ is a real number for all $x, y \in M$.
\end{lemma}
\begin{proof}
    Since the operators $xy$ and $yx$ are isospectral, i.e.~they have the same non-vanishing eigenvalues, implies that
    if $\lambda^{xy}$ is an eigenvalue of $A_{xy}$ also $\complexconj{\lambda^{xy}}$ is an eigenvalue of $A_{yx}$.
    Thus, $\frac{1}{2n} \tr_{\H} {\tilde{A}_{xy}}$ as the sum over the eigenvalues of $A_{xy}$ is a real number.
\end{proof}

This completes our discussion of causal correlations in the abstract setting.
In the next section, we analyze the defined objects in the continuum limit.

\section{Continuum Limit}\label{sec:continuum-limit}

In this section, we discuss the continuum limit of \ac{cfs}.
We focus only on the relevant aspects for this paper.
For a complete derivation of the forces of the standard model and the derivation of general relativity, see~\cite{cfs}.

Note that in our reading of a \ac{cfs} as a state of a $N$ particle many-body quantum system, the continuum limit corresponds
to the large $N$ limit; hence $N \to \infty$.
The idea behind the construction of the local correlation map was introduced pedagogically in~\cite{mmt-cfs, Finster2024}.
The key insight here is that thanks to the local correlation map and the realization of the Hilbert space in terms of
physical wave functions, we can work with conventional function spaces in Lorentzian manifolds to construct a \ac{cfs}.
For simplicity, we restrict ourselves to the case of a Minkowski spacetime as the classical spacetime.

Thus, identifying Minkowski space with~$\mathcal{M} = \R^{1,3}$, endowed with the standard Minkowski inner product with signature
convention $(+,-,-,-)$, and a fixed reference frame.
We consider smooth solutions of the Dirac equation in Minkowski space,
\begin{align*}
    \left(i \gamma^k \partial_k - m\right) u = 0,
\end{align*}
with spatially compact support $u \in C_{\text{sc}}^\infty(\R^{1,3})$ (i.e.~with compact support in the spatial variables).
The space of solutions is endowed with the usual scalar product
\begin{align*}
    \braket{u}{v} \defeq \int_{t=\text{const}} \diracad{u(t, \vec{x})} \gamma^0 v(t, \vec{x}) \dd\vec{x},
\end{align*}
where $\diracad{u} = u^\dagger\gamma^0$ is the adjoint spinor.

The completion of the solution space with respect to the scalar product defines the Hilbert space.
However, occupying all states in this Hilbert space does not correspond to a minimizing \ac{cfs}.
Thus, the Hilbert space of all solutions gives the auxiliary Hilbert space $\tilde{\H}$.
For the vacuum, only the negative energy solutions of the Dirac equation are occupied.
This defines the Hilbert space $\H$.

To construct the continuum limit for Minkowski space~\cite[Remark 1.2.1 and Section 2.4]{cfs} we want to work with a sequence of \ac{cfs} which minimize
the causal action principle and limit to the Dirac sea in a suitable sense.
The causal action principle is known to be well-defined for finite-dimensional Hilbert spaces~\cite[Chapter 12]{Bernard2014, Finster2024}.
Starting with the Dirac sea we face two problems: first the spectrum, i.e.~the mass-shell is a continuum and second, it is non-compact.
To discretize the spectrum, we need to introduce an IR-cutoff in the form of a large but finite spatial box of side length $L$.
To compactify the spectrum, we need to introduce a regularization in terms of a UV-cutoff on the order of the Planck energy $\varepsilon^{-1}$.
This leaves us with an $N$-dimensional subspace $\H_N$ of negative energy solutions where the dimension $N(L,\varepsilon)= \dim \H_N$
is a function of the box size $L$ and the energy cutoff $\varepsilon^{-1}$.
It is clear that $N(L,\epsilon)\rightarrow \infty$ for both $\varepsilon$ fix and $L\rightarrow \infty$ as it is de-discretizes the spectrum, as well as $L$ fix and $\varepsilon \rightarrow 0$
as it de-compactifies the spectrum.
The continuum limit of Minkowski space is obtained in the limit where we take both $L\rightarrow \infty$ and $\varepsilon \rightarrow 0$ simultaneously giving us the full subspace of all negative-energy
solutions, i.e.~the Dirac sea\footnote{
    The box size $L$ can be thought of as the cosmological horizon scale $r_{\mathcal{CH}}= \left( \frac{3}{\Lambda}\right)^{1/2}$
    and hence $\Lambda$ as an infrared cut-off.
    If we assume that the cosmological constant depends on the regularisation scale $\varepsilon$ such that $\Lambda(\varepsilon) \searrow 0$
    for $\varepsilon \searrow 0$ it should be possible to construct a continuum limit for Minkowski space from the interacting
    region of de-Sitter space.
    In this case $N(\Lambda(\varepsilon),\varepsilon)$  would depend on the regularization length $\varepsilon$ alone and
    we would have a simultaneous UV-IR-limit.
    Note that according to \cite{visser2002sakharov} such a dependence of the cosmological constant on the Planck scale was already proposed
    by Sakharov in \cite{sakharov2000vacuum}.
}.

Note that to encode the full Minkowski spacetime we need an infinite-dimensional Hilbert space, and thereby we do not run
into Shannon's sampling theorem here.
However, if we were to extend the Hilbert space to include states of the auxiliary Hilbert space $\tilde{\H}$, we would
not add additional information.
A \ac{cfs} that includes positive energy states as physical wave functions does not encode a vacuum spacetime anymore.
However, the additional positive frequency states will appear as fermionic matter in the effective continuum description.
This is the fundamental insight that underlies the baryogenesis mechanism derived from \ac{cfs} in~\cite{baryogenesis,finster2025baryogenesis}. 

Now we construct the local correlation map for this setting explicitly, starting from the finite-dimensional Hilbert
space $\H_N$ with the additional assumption that the $N$ wave functions we choose are all continuous.
In this case, the point-wise inner product of the wave
functions defines a bounded sesquilinear form~$\spinbraket{\cdot}{\cdot}_x$ on $\H$,
\begin{align*}
    \spinbraket{u}{v}_x \defeq - \diracad{u(x)} v(x).
\end{align*}
Riesz representation theorem allows us to represent the sesquilinear form by a uniquely determined bounded operator on
the Hilbert space, which we denote by~$F(x)$, i.e.
\begin{align*}
    \braket{u}{F(x) v}_\H \defeq \spinbraket{u}{v}_x && \text{ for all } \ket{u}, \ket{v} \in \H.
\end{align*}
By construction, this operator is self-adjoint, has rank at most four, and has
at most two positive and two negative eigenvalues.
This means that $F(x)$ is a map from Minkowski space into the set $\F_2$, hence the name ``local correlation map''.

\subsection{Causal Correlations in Minkowski Spacetime}\label{subsec:causal-correlations-in-the-continuum}

Using this representation for the Hilbert space $\H$ and the local correlation map $F(x)$, we compute the kernel of the
fermionic projector according to equation~\eqref{eq:kernel-of-the-fermionic-projector}.
This gives (for details, see~\cite[Chapter~1]{cfs} or~\cite[Section~5.6]{Finster2024})
\begin{align}
    \label{eq:p-sea}
    P(x, y) &= \int \frac{\dd^4 k}{(2 \pi)^4} \left(\slashed{k} + m \right) \delta(k^2 - m^2) \Theta(-k^0) e^{-ik(x-y)} \\
    &= i \alpha(x, y) (\slashed{y} - \slashed{x})  + \beta(x, y)
\end{align}
where
\begin{align} \label{eq:alpha-beta}
    \alpha(x, y) &\defeq \frac{m^4}{8 \pi^3} \frac{K_2(z)}{z^2}, & \beta(x, y) &\defeq \frac{m^3}{8 \pi^3} \frac{K_1(z)}{z}.
\end{align}
Here $K_1$ and $K_2$ are the modified Bessel functions of the second kind, and their argument is defined as $z \defeq m \sqrt {-(y - x)^2}$.
Note that $P(x, y)$ in the case of the Minkowski vacuum is, up to a factor of $2 \pi i$, equal to the standard
two-point correlation function of a free fermionic field~\cite[Chapter 3]{Peskin2018}.
Thus, it is natural to think of $A_{xy} = P(x, y) P(y, x)$ as a sort of two-point correlation strength operator.
Using the concrete form of $P(x, y)$ we can compute $A_{xy}$ as
\begin{align}
    A_{xy} = a(x, y) (y_\mu - x_\mu) \gamma^\mu + b(x, y)
\end{align}
with
\begin{align}
    a(x, y) &\defeq 2 \Im \left( \overline{\alpha(x, y)} \beta(x, y) \right), \\
    b(x, y) &= \abs{\alpha(x, y)}^2 (y - x)^2 + \abs{\beta(x, y)}^2.
\end{align}

From here, we directly obtain the eigenvalues of the closed chain by
\begin{align}
    \lambda^{xy}_{\pm} = b(x, y) \pm a(x, y) \sqrt {(y - x)^2}.
\end{align}
In the continuum limit, we recover that the causal structure of \ac{cfs} given by Definition~\ref{def:causalstructure}
coincides with the usual causal structure of the Minkowski spacetime.
\subsection{Causal Action as a Variance of Causal Correlations}\label{subsec:causal-action-as-a-variance-of-correlations}

In general, the kernel of the fermionic projector can be any symmetric operator from $S_y$ to $S_x$ (symmetric in the sense of
proposition~\ref{prop:properties-of-closed-chain}).
For the Minkowski vacuum, $P(x, y)$ given by equation~\eqref{eq:p-sea} has a particular vector, scalar structure,
i.e.~in the vectorspace of symmetric operators from $S_x$ to $S_y$ with basis $(\mathbb{1}, i \gamma^5, \gamma^\mu, \gamma^5 \gamma^\mu, \Sigma^{\mu\nu})$\footnote{
    Here, we have $\gamma^5 = i \gamma^0 \gamma^1 \gamma^2 \gamma^3$ as the usual fith gamma matrix and we refer to $i \gamma^5$
    as the pseudo scalar component.
    In addition, we have the bi-linear Lorentz boost tensor $\Sigma^{\mu\nu} = \frac{i}{2} \left[ \gamma^\mu, \gamma^\nu \right]$.
}
only the coefficients of $\mathbb{1}$ and $\gamma^\mu$ are non-zero.
In this case, the closed chain also has this vector, scalar structure.

If in addition, we consider the more general case where
\begin{align} \label{eq:general-vector-scalar-structure}
    P(x, y) = i \alpha(x, y) \slashed{\xi} + \beta(x, y)
\end{align}
with general complex functions $\alpha$, $\beta$ and a complex vector $\xi_\mu(x, y)$, then the closed chain also has
bilinear components, i.e.
\begin{align} \label{eq:closed-chain-vector-scalar-decomposition}
    A_{xy} = b(x, y) + a_\mu(x, y) \gamma^\mu + c_{\mu\nu}(x, y) \Sigma^{\mu\nu}.
\end{align}
For example, this is the case for a regularized Minkowski vacuum.
Under these assumptions, it is possible to rewrite the Lagrangian in terms of a variance related to the expectation value
the causal correlation operator and the two-point-correlation strength.

\begin{proposition}
    Let $P(x, y)$ be the kernel of the fermionic projector of the form~\eqref{eq:general-vector-scalar-structure}, then
    the Lagrangian is given by
    \begin{align}
        \L(x, y) = 4 \Var_\Omega \left[ \tilde{A}_{xy} \right] \defeq 4 \expval{\left( \tilde{A}_{xy} - \expval{\tilde{A}_{xy}} \right)^2}_{\Omega}
    \end{align}
\end{proposition}
\begin{proof}
    As stated above, if $P(x, y)$ has the form~\eqref{eq:general-vector-scalar-structure}, then the closed chain is given by
    equation~\eqref{eq:closed-chain-vector-scalar-decomposition}.
    Thus, we have
    \begin{align*}
        \left( A_{xy} - b \right)^2 = a_\mu a^\mu + 2 c_{\mu\nu} c^{\mu\nu}
    \end{align*}
    and hence the eigenvalues of $A_{xy}$ are given by
    \begin{align*}
        \lambda^{xy}_{\pm} = b \pm \sqrt{a_\mu a^\mu + 2 c_{\mu\nu} c^{\mu\nu}}
    \end{align*}
    each with two-fold degeneracy.
    Since $\gamma^\mu$ and $\Sigma^{\mu\nu}$ are trace free, we have
    \begin{align*}
        b &= \frac{1}{4} \tr_{S_x} \left[ A_{xy} \right] = \expval{\tilde{A}_{xy}}_{\Omega}, \\
        a_\mu a^\mu + 2 c_{\mu\nu} c^{\mu\nu} &= \frac{1}{4} \tr_{S_x} \left[ \left( A_{xy} - b \right)^2 \right] = \Var_\Omega \left[ \tilde{A}_{xy} \right].
    \end{align*}
    For non-spacelike separated points, we have
    \begin{align*}
        \L(x, y) &= \frac{1}{8} \sum_{i, j = 1}^{4} \left( \abs{\lambda^{xy}_i} - \abs{\lambda^{xy}_j} \right)^2 \\
        &= \frac{1}{2} \sum_{i, j = \pm} \left( \abs{\lambda^{xy}_i} - \abs{\lambda^{xy}_j} \right)^2 \\
        &= \left( \abs{\lambda^{xy}_+} - \abs{\lambda^{xy}_-} \right)^2 = a_\mu a^\mu + 2 c_{\mu\nu} c^{\mu\nu} = 4 \Var_\Omega \left[ \tilde{A}_{xy} \right].
    \end{align*}
    where we used that $\lambda^{xy}_+ \lambda^{xy}_- \geq 0$.
    For spacelike separated points, the Lagrangian vanishes.
\end{proof}

This proposition gives a nice interpretation of the Lagrangian in the continuum limit as a variance of the causal correlation operator.
Treating the expectation value of the causal correlation operator as a two-point correlation strength, the causal action minimizes
the fluctuations of the correlation strength of all causally separated points.

In addition, the interpretation of the Lagrangian as a variance of the causal correlation operator gives a good understanding
of the `emergence' of the classical spacetime.
In the next section, we present some numerical analysis of these structures for the regularized Minkowski vacuum.
The results indicate that in the limit of vanishing regularization the total variance tends to zero.
Thus, the classical spacetime is the limit of equipartition of the correlation strength.

\subsection{Numerical Computations}\label{subsec:numerical-computations}

In the continuum limit, the kernel of the fermionic projector~\eqref{eq:p-sea}
is divergent on the lightcone.
One way of regularizing is to exponentially suppress the high energy modes by introducing a UV-cutoff scale $\varepsilon$.
Then,
\begin{align}
    P(x, y) = \int \frac{\dd^4 k}{(2 \pi)^4} (\slashed{k} + m) \delta(k^2 + m^2) \Theta(-k_0) e^{k_0 \epsilon} e^{- i k (x - y)}.
\end{align}
We refer to this as the $i\varepsilon$-regularized Minkowski spacetime.
Performing the integration gives
\begin{align}
    P(x, y) = i \alpha \slashed{\xi} + \beta
\end{align}
with $\xi^\mu = (y^0 - x^0 - i\varepsilon, \vec{y} - \vec{x})$ and $\alpha, \beta$ given by equation~\eqref{eq:alpha-beta},
but the argument of the Bessel function changed to $z \defeq m \sqrt {- \xi_\mu \xi^\mu}$.
Note that for $\varepsilon > 0$ the functions $\alpha, \beta$ are well-defined for all $x$ and $y$.
The closed chain is then given by equation~\eqref{eq:closed-chain-vector-scalar-decomposition} with non-vanishing coefficients
\begin{align}
    b &= \abs{\alpha}^2 (t^2 - r^2 + \varepsilon^2 + \abs{\frac{\beta}{\alpha}}^2), \\
    a_{\mu} &= 2 \abs{\alpha}^2 \Im \left[ \frac{\beta}{\alpha} \complexconj{\xi_{\mu}} \right] \eqdef 2 \abs{\alpha}^2 X_\mu, \\
    c_{0i} &= - 2 \abs{\alpha}^2 \varepsilon (y_i - x_i).
\end{align}
Thus, we find that two points $x$ and $y$ are timelike separated if $X_\mu X^\mu > \varepsilon^2 r^2$.
In this case, we have
\begin{align}
    \L(x, y) = 4 \Var \left[ \tilde{A}_{xy} \right] = 16 \abs{\alpha}^2 \left( X_\mu X^\mu - \epsilon^2 r^2 \right).
\end{align}

As stated in Section~\ref{sec:preliminaries} the causal action is minimal if the $\ell(x)$ (equation~\eqref{eq:el-equation-l})
is minimal and constant for all $x \in M$.
In particular, we choose the Lagrange multiplier $\s$ of the volume constraint such that $\ell(x) = 0$ for all $x \in M$.

We now want to analyze the contribution of the integral over $\L$ to $\ell(x)$.
This part corresponds to the total variance of the correlations from $x$ to all other spacetime points.
\begin{align}
    \int_\mathcal{F} \L(x, y) \dd \rho(y) &= 4 \int_{\mathcal{F}} \Var \left[ \tilde{A}_{xy} \right] \dd \rho(y) \nonumber \\
    &= 16 \int_{X_\mu X^\mu > \varepsilon^2 r^2} \abs{\alpha}^2 \left( X_\mu X^\mu - \epsilon^2 r^2 \right) \dd^4 y.
\end{align}

At this point, we want to take the limit $\varepsilon \to 0$ and analyze the behavior of the integral.
However, since $\varepsilon$ is a dimensionful parameter, we have to compare it with the other terms of $\ell(x)$.
In the following analysis, we neglect the $\kappa$ term.
In addition, we have to fix the two scaling degrees of freedom~\cite{Curiel2020}, which come from the invariance under $\rho(\Omega) \to \sigma \rho \left( \frac{\Omega}{\lambda} \right)$.
First, we restrict the spacetime to a box with edge length $L$ and choose $\sigma$ such that the total volume is $1$.
The second scaling freedom is fixed by choosing the local trace as $1$, i.e.
\begin{align}
    1 \stackrel{!}{=} \tr P(x, x) \approx \frac{\lambda}{(2 \pi)^3} \frac{m}{\varepsilon^2} \Rightarrow \lambda = (2 \pi)^3 \frac{\varepsilon^2}{m}.
\end{align}
Having fixed the scaling parameters, we have
\begin{align}\label{eq:total-variance}
    l_\epsilon(x) \defeq \int \L(x, y) \dd \rho(y) = 16 \lambda^4 \int_{\substack{X_\mu X^\mu > \varepsilon^2 r^2 \\ \abs{t}, r < L}} \abs{\alpha}^2 \left( X_\mu X^\mu - \epsilon^2 r^2 \right) \dd^4 y.
\end{align}
Note that the integrand only depends on the difference $\xi$.
Thus, $l_\epsilon(x)$ is independent of $x$ and due to the fixing of the scaling parameters, we now can take the limit
$\varepsilon \to 0$.
We compute $l_\epsilon(x)$ for different values of $m \varepsilon$ numerically.
The results are shown in Figure~\ref{fig:total-variance}.
Fitting the numerical values gives rise to a power law of $\varepsilon^8$ which in limit $\varepsilon \to 0$ goes to zero.

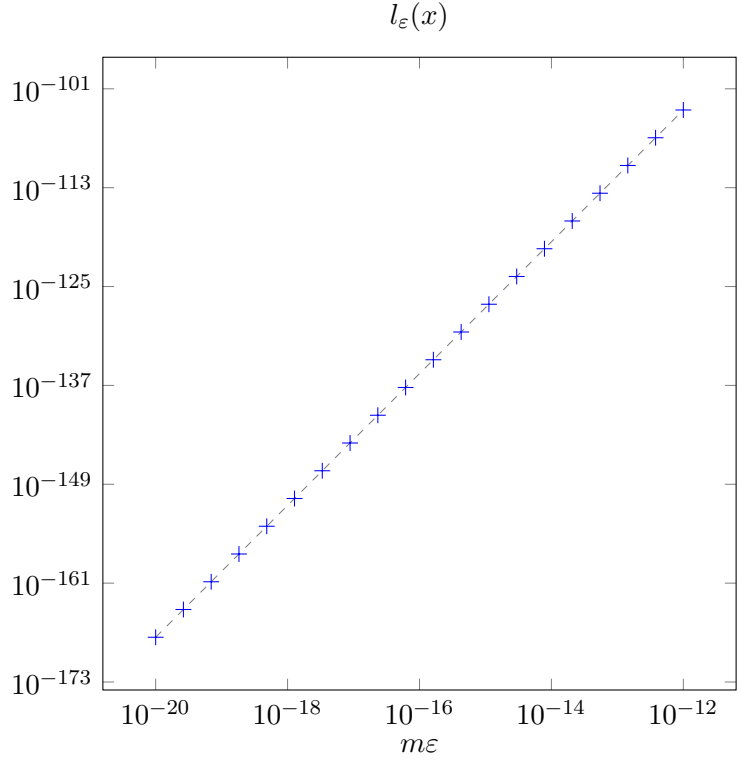
\begin{figure}
    \label{fig:total-variance}
 \begin{tikzpicture}
    \begin{axis}
        [
        xlabel={$m\varepsilon$},
        xtick={1e-20, 1e-18, 1e-16, 1e-14, 1e-12},
        xmode=log,
        ymode=log,
        title={$l_\varepsilon(x)$},
        width=10cm,
        height=10cm
        ]
        \addplot[
            blue,
            only marks,
            mark=+,
            mark size=3pt
        ] table {numerical-integration.dat};

        \addplot[domain=1e-20:1e-12, samples=100, dashed, gray] {2.70e-08 * x^8};
    \end{axis}
\end{tikzpicture}
    
    \caption{
        The plot shows the numerical integration of expression of equation~\eqref{eq:total-variance} for different values
        of $m \varepsilon$.
        The numerical values are fitted using the function $\varepsilon \mapsto a \varepsilon^b$ with $a = 2.7 \cdot 10^{-8}$ and $b = 8$.
    }
\end{figure}

\section{Description of Experiments}\label{sec:experiment}

With the pieces in place, it is now time to explain how to interpret experiments within the setting of \ac{cfs}.
First, recall that fundamentally \ac{cfs} is a completely relational theory, where the causal action principle determines
which correlations between the fundamental fermionic degrees of freedom are realized in a physical system.
This web of correlations is the only thing that makes up spacetime in perfect alignment with Mach's principle in the sense
that if you take away all matter, i.e.~all the fermionic states\footnote{
    The fermionic states need not manifest in spacetime as matter the way we conventionally define it.
    This will be elaborated later on.
}, there is no more space as there is nothing to relate.
In fact, for a \ac{cfs} with spin dimension $n$ for there to exist any regular spacetime points, the dimension of the
Hilbert space must be at least $2n$.

Now, there are two crucial questions that we need to answer before we can combine them to form a full description of our everyday
physical experience:
\begin{enumerate}
    \item How do you test a system that comprises the entire universe from within itself?
    \item How does a purely relational theory give rise to our conventional experience of time and space?
\end{enumerate}

To answer the first question, we take inspiration from \ac{gr} which is faced with a similar challenge.
If everything curves spacetime, how do you test the curvature of spacetime?
The answer here is simple: Take a system that is small enough such that its effects on spacetime are negligible
in terms of its ``mass'', i.e.~its contribution to spacetime curvature, compared to the contributions of all other systems.
Then one can treat the ``test body'' as a geodesic in spacetime and study its trajectory.
Now in the \ac{cfs} setting, we have fewer structures to work with.
To make progress at this point, we make use of the Definition~\ref{def:subsystem} of a subsystem.
We split the total physical system $(\H, \Omega)$ into a probe $(\H_P, \omega_P)$ and the background $(\H_B, \omega_B)$ such that\footnote{
    Many attempts to define concepts of locations starting with an abstract Hilbert space, rely on tensor product structures (see e.g.~\cite{niederklapfer2024fundamental,franzmann2024or,zanardi2004quantum}).
    Here the split into background and probe and the emergence of spacetime happen in separate steps.
    Accordingly a direct product is much more natural.
    The analogy with solid state theory can be helpful again. There exists a larger underlying Hilbert space given by the total band structure of the material. However, at zero temperature only a subspace thereof is occupied such that there are $N$ band electrons per lattice cell. The subspace consists of all the one particle states in tha band structure that are occupied. 
    Furthermore, when we work with band structures in a lattice, it is still completely natural to consider occupied states in the band
    as individual electrons that can travel through the lattice and who's motion we can in principle study individually.
    In that setting it is also natural to consider subspaces when describing individual electrons rather than tensorial factor.
}

\begin{align}
    \H &= \H_P \oplus \H_B &\text{ and }&& \Omega &= \omega_P +  \omega_B.
\end{align}
In a typical experiment, we then require that the number of degrees of freedom is several orders of magnitude smaller than
the number of degrees of freedom in the background, i.e.~$\dim (\H_P) \ll \dim (\H_B)$.
The equivalent of geodesics in \ac{gr} in our context would be the idealized probe consisting of a single fermionic
degree of freedom.

To answer the second question, we have to join together a few key ideas.
First, we observe that the continuum limit for \ac{cfs} as first derived in~\cite{cfs} and explained in Section~\ref{sec:continuum-limit}
corresponds to the limit of infinite states, i.e.~$N = \dim \H \rightarrow \infty$.
The local correlation map~\eqref{eq:localcorrelation} $ F[g_{\mu\nu}, A_\mu,\dots]: \mathcal{M} \rightarrow \F_n$ provides us with a macroscopic
description of the \ac{cfs} in terms of emergent variables $g_{\mu\nu}, A_\mu, \dots$ akin to temperature, pressure and
so on in thermodynamics.
Hence, the classical fields defined on a spacetime are considered as statistical quantities derived from an averaging
process over a large number of degrees of freedom.
The causal action principle takes the role of the law of equipartition in thermodynamics.
This can be seen from the fact that in certain important examples (see Section~\ref{subsec:numerical-computations})
the causal action reduces to principle of minimal fluctuations.
This suggests that in the large $N = \dim\H$ limit a large class of \ac{cfs}
admit an emergent description in terms of a local correlation map and therefore in terms of classical field $g_{\mu \nu}, A_\mu$, etc.
The large $N$ limit of the background system with its description in terms of these fields then gives rise to the
physical concepts we are familiar with\footnote{
    At this point it is worth mentioning that on a universal scale it is precisely the CMB, a phenomena we study entirely
    in terms of its correlations, that allows us to establish a notion of absolute speed in the universe.
}.

To characterize our physical system, we then study the evolution of a probe subsystem $(\H_P, \omega_P)$ with respect to the
web of correlations spanned by the background subsystem $(\H_B, \omega_B)$.
While in principle, we can treat any subspace of $\H$ as the probe, in practice, to perform an actual experiment two caveats apply:
\begin{enumerate}
    \item The web of correlation spanned by the background subsystem $(\H_B, \omega_B)$ must admit an approximate effective description in terms of a continuum limit.
    \item To be able to study the evolution of a system in a laboratory, we have to be able to meaningfully localize the probe on the scale of the experiment.
\end{enumerate}

The first condition can always be satisfied when we construct a \ac{cfs} from a continuum description with a suitable regularization. 
For the second condition, our notion of localization of states introduced in Section~\ref{subsec:spacetime-superposition} is crucial.
In particular, we need to be able to localize a state $u_i$ within spacetime regions $\mathcal{U}_{i} \subset \mathcal{U}_L \subset M$
where $\mathcal{U}_L $ is the spacetime region covered by the experiment.
The collection of spacetime regions $\mathcal{U}_{i}$ are then our data points.

Regarding the example of Minkowski space, we have the following corollary. 
\begin{corollary}
      All states in the continuum limit of Minkowski space are delocalized. 
\end{corollary}
  
\begin{proof}
This follows immediately from the fact mentioned in Section~\ref{sec:continuum-limit} that the Hilbert space in the continuum limit only contains negative energy solutions and the result of Hegerfeldt~\cite{hegerfeldt1974remark} and the quantitative analysis in~\cite{finster2023incompatibility} that solutions with support in a semi-bounded spectrum cannot have spatially compact support.    
\end{proof}

As a result, Minkowski space, as an effective description, appears as a vacuum spacetime where no experiments can take place. 
This confirms a remark in~\cite{isidro2024gravitation}, that the delocalized states play a key role in the description of experiments.
Namely, it is those states that make up the background, and only once we perturb Minkowski space by adding states that are localized in a meaningful way, we arrive at
a continuum limit description containing matter.
In an actual experiment, we track the localized states to probe the web of correlations
spanned by the background.

\subsection{(Quantum) Reference Frames in CFS}\label{subsec:reference frames}

At the fundamental level \ac{cfs} is a Hilbert space-based relational theory.
Yet, the probe-background split, together with the local correlation map,
provides us with a reference frame, which we can use to describe experiments.
An immediate question that arises is how this notion of reference frame compares to ideas discussed in the \ac{qrf} community,
see, e.g.~\cite{lake2023quantum,goeller2022frames, quantumFrameRelativity, bartlett2007reference, fewster2025quantum,Apadula2024quantumreference, ali2024quantum, Krumm2021quantumreference,Yang2020switchingquantum,carette2025operational,giacomini2019quantum}.
We will now comment on some aspects.

\subsubsection{Preferred Reference Frame}\label{subsubsec:preferedRF}
First we note that there is, at least heuristically, a preferred probe-background split.
Given that the split is given by a direct sum,
we can consider the set of all bases' $\mathcal{B} \defeq \{(u_i) | (u_i) \text{ basis of } \H \}$.
The number of delocalized states depends on the choice of a basis $(u_i)$.
A preferred basis for $\H$ is one that maximizes the number of delocalized states.
For a given basis $(u_i)$, the background $(\H_B,\omega_B)$ is chosen as the subspace of all delocalized states,
i.e.~$\H_B = \Span\{u_i | u_i \text{ delocalized}\}$.
The remaining states can be split into those
making up the instruments $(\H_I,\omega_I)$ and the actual probe $(\H_P,\omega_P)$
studied in the experiment\footnote{If we consider the universe as a whole then for most experiments there will be another sector of the Hilbert space $(\H_R,\omega_R)$ comprised of the states localized in spacetime regions disjoint from the laboratory $\mathcal{U}_L$. In practice, these states can be simply dropped from the \ac{cfs} to describe the specific experiment at hand.}.
In practice, the condition for the background states to be delocalized can actually be weakened.\\

\textit{
We say that a state $u_i$ is quasi-delocalized with respect to a spacetime region $\mathcal{U_L}$ if $\mathcal{U_L}\subset M_i$ and if for all spacetime regions $\mathcal{U}$ containing $\mathcal{U_L}\subset \mathcal{U}$ with $u_i$ localized in $\mathcal{U}$ we find that $\mathcal{U}$ is substantially larger than $\mathcal{U}_L$.
} \\

We can weaken the condition
under which we consider a state
to be part of the background by $\H_B = \Span\{u_i|u_i \text{ quasi-delocalized with respect to }\mathcal{U}_L\} $ accordingly.
Once we have established a preferred split, the choice of a basis in $\H_B$ remains.
Such a change of basis corresponds to a change in the reference frame of the background system.
It is important to remember that in the approximation of the effective continuum description,
every basis vector corresponds to a physical wave function
and that the local correlation map $F[g_{\mu\nu}, A_\mu,\dots]: \mathcal{M} \rightarrow \F_n$ is built from the correlations of these wave functions.
As a result, there can in principle be multiple continuum approximations for the same \ac{cfs},
as long as they give rise to the same web of correlations of the physical wave functions.
Ockham's razor then compels us to choose the simplest description of the continuum approximation that fits the bill.

\subsubsection{Superposition of Macroscopic Objects}
To sharpen our understanding of the nature of background reference frames in \ac{cfs}, we want to focus on a set of
ideas that have been studied in the context of spacetime superposition such as the superposition of
causal structures induced by, e.g.~earth in spatial superposition (see for example~\cite{zych2019bell}).
We have already seen in Corollary~\ref{cor:fixcausalstructure} that the nature of the causal relation between two points
in a \ac{cfs} is independent of the subsystem spacetime within which they are studied\footnote{
    The distance between two points is generally different for different subsystem spacetimes.
}.
We now explain how on the level of the effective description in terms of classical spacetime a superposition of macroscopic
objects and thereby of spacetime causal structure cannot arise in the context of \ac{cfs}.
In particular, we demonstrate that there is no meaningful notion of a ``reference frame'' of the probe $(\H_P, \omega_P)$.

First of all, we note that in \ac{cfs} we cannot have a physical system
that consists only of the earth $(\H_E, \omega_E)$ and a quantum system $(\H_P, \omega_P)$.
Instead, for there to be a meaningful notion of spatial relation between the systems $E$ and $P$ in \ac{cfs} we require
a background system $(\H_B, \omega_B)$ with $\dim \H_B \gg \dim \H_E + \dim \H_P$ and
\begin{align}
    \H = \H_P \oplus \H_E \oplus \H_b.
\end{align}
Recall that the variables of the continuum limit are statistical averages over the background states.
Ignoring for a moment the question of suitable (de)-localization a comparison of dimensions is sufficient to convince oneself
that there does not exist such a thing as the ``reference frame'' of the system $P$.

The reference frame in which we perceive the probe to be in spatial superposition is given by an effective description of $\H_E \oplus \H_B$.
Because the effective description of the background is given by a statistical average over all background states,
changing the reference system to $\H_P \oplus \H_B$ barely affects the effective description.
As a result, the Earth would not appear to be in superposition in that new reference frame,
unless it was already in the reference frame obtained from $\H_B$ alone.
Due to the fact that $\dim \H_B \gg \dim \H_P$ the contribution to the effective description from the probe is simply negligible.
Here we recall again Corollary~\ref{cor:fixcausalstructure} which shows
that the causal relation between any two points $x, y$ in spacetime $M$ is independent of the choice of the probe-background split,
and hence no admissible continuum description of a \ac{cfs} can contain a superposition of macroscopic causal structures.

If instead we  try to designate the space $\H_P$ as our background, then we run into
two problems: First, the system $\H_P$ was implicitly assumed to be localized (otherwise it would make no sense to be
speaking about it being in spatial superposition\footnote{
    Here spatial superposition refers to the idea that we can localize the state $u_i$ in a spacetime region
    $\mathcal{U} = \mathcal{U}_1 \cup \mathcal{U}_2$ with all spacetime points in $\mathcal{U}_1$ being spatially separated
    from all spacetime points in $\mathcal{U}_2$.
}) and therefore holds no information about the distant correlation structure of the \ac{cfs}.
Second, we assumed that the probe was small, in the most extreme case even consisting of just a single degree of freedom,
and hence it simply does not hold enough information to admit a viable continuum approximation\footnote{
    This is consistent with the conclusion in~\cite{reitz2023model} that not all webs of correlations admit a viable continuum approximation.
}.
This conceptual disagreement between the scenario considered for example in~\cite{zych2019bell} and \ac{cfs} might allow
for an experimental distinction between these two sets of ideas.

\section{Outlook and Conclusion}\label{sec:outlook}
In the present paper we put forward a first proposal for a full CFS ontology that works within the structures of the theory alone, relying on the continuum limit only as an efficient tool to encode the information contained in the web of correlations that makes up the \ac{cfs}. This is a significant step forward in establishing \ac{cfs} theory as a leading contender to describe our physical world at all scales.

However, many of the ideas we introduced are still rough on the edges. In the following, we present a necessarily incomplete list of additional questions that arise from the present work.
\begin{enumerate}
      \item In the course of our paper we have drawn inspiration from a wide range of different schools of thought in fundamental theoretical physics. This led us to introduce a variety of new concepts to the study of \ac{cfs}. Foremost, we introduced the idea of conceiving a \ac{cfs} as a many-body quantum system in which concepts can be defined in such a way that they make sense for a subsystem and can be added up to obtain the respective properties for the total system. In that spirit, we introduced the position observables for \ac{cfs} and the associated notion of one-particle spacetimes that can be superposed to obtain the total spacetime. The one exception to this linear superposition is the Lagrangian, which we have shown in an important subclass of \ac{cfs} to be the variance of the total correlation strength and thereby manifestly nonlinear in the one-particle states. 

    As a result, when studying linear perturbations of a \ac{cfs} we can treat the full nonlinear minimizer as a background on which the perturbations propagate\footnote{This is conceptually similar in the context of linearized gravity.}. All notions of observables and localization keep making sense for these linearized perturbations. In many situations where the number of degrees of freedom of the probe $n=\dim \H_p$ is much smaller than the dimension of the total Hilbert space $n \ll N= \dim \H$ one treats the subsystem as a linear perturbation. Currently, it needs to be shown on a case-by-case basis that this approximation is valid. It would be worthwhile to study whether the validity of the linear approximation can be established more broadly under the assumption that $\lim \frac{n}{N}\rightarrow 0$. 
    
    \item In Definition~\ref{def:1Pcorrelation} we introduced the one-particle two-point correlation strength $b_{v_i}(x,y)$. In the spirit of the ideas in~\cite{kempf2021correlational, reitz2023model} it is tempting to use this notion of correlation strength to try to define a notion of distance in the spacetimes $M_i$ and the total spacetime $M$. This would give rise to a well-defined correlation geometry on $\F$ and would allow us to interpret the action principle as minimizing the fluctuations in the causal distance across one-particle spacetimes.   
    
    The discussion in Section \ref{subsec:spacetime-superposition} then provides us with a clear framework for the superposition of such correlation geometries, and in the continuum limit the fluctuations go to zero. This would provide a physical explanation for the well-known result~\cite[Section 1.2.5]{cfs} that the causal structure becomes transitive, i.e.~completely well behaved/classical in the continuum limit.
    At present, however, it is unclear to the authors how to define such a notion of distance in a way that it has all the desired properties.

    \item In~\cite{dappiaggi2018linearized,finster2022construction} it was shown that a transitive causal structure can be obtained at mesoscopic scales by removing the contributions from the high-frequency states. Given that the notion of localization we introduced in Section~\ref{subsec:spacetime-superposition} deals with spacetime regions $\mathcal{U}\subset M\subset \F$, it would be interesting whether it might be possible to establish similar transitive, mesoscopic causality relations working with the intrinsic structures of \ac{cfs} alone. Conceptually what happens when we remove the high frequency states is that we only consider correlation information arising from states that vary slowly relative to the size of the spacetime regions who's causal relation we want to determine. Now on the fundamental level it is not a priori clear how to define relative scales to determine which states should be omitted. Here, the above-mentioned notion of distance might play a crucial role. The hope is that for any number of spacetime regions large enough, one might be able to derive a notion of causal order that is transitive.

    \item In Section \ref{subsec:reference frames}, we introduced the idea of a background and an associated reference frame.  
    A key result in the present work is the fact that the \ac{cfs} framework is conceptually incompatible with some scenarios studied in the literature, e.g.,~\cite{zych2019bell}, because in \ac{cfs} the vast majority of states are assigned to the background -- i.e., reference -- system, and Corollary~\ref{cor:fixcausalstructure} guarantees that the nature of the causal relation between two points is independent of the subsystem considered. This might open a pathway for an experimental distinction between the two frameworks.
    
    However, a correspondence with the more conservative aspects of \ac{qrf} seems plausible, since both study the evolution of a quantum system using other quantum states as a reference. For example, it would be interesting to see whether our background states can be understood as the ``frame states'' that ``dress'' the observable in the sense of~\cite{goeller2022frames}. It would also be good to understand whether the idea of quantum frame relativity~\cite{quantumFrameRelativity} can be adapted to our context.
    
    \item In Section \ref{subsubsec:preferedRF} we introduced the concept of quasi-delocalized states. This concept might be crucial to solving one of the longest-standing challenges in physics: determining the nature of dark matter. Based on the novel mechanism of baryogenesis introduced in~\cite{baryogenesis} and refined in~\cite{finster2025baryogenesis}, in~\cite{paganini2020proposal} a new storyline for the universe was proposed, which naturally leads to the idea of \ac{fcdm}. In this model the fermionic dark matter particles are in the lowest kinetic energy state available. In a galaxy that implies that the potential well is filled with stated from the bottom up. A key point about this kinetic configuration is that the states of the individual particles are only localized on scales comparable to the galaxy itself. This opens up an intriguing possibility. For physical processes that happen at sub-galactic scales, e.g. the orbits of stars and planets, the \ac{fcdm} states are quasi-delocalized and can therefore equally well be considered as part of the background or as part of the matter sector. Given the discussion at the end of Section~\ref{subsec:observables} we can think of these states, when assigned to the background, as providing an extra contribution to the volume form. Thinking of the volume form as a density of states along the lines of~\cite{isidro2024gravitation} this would lead to a background on galactic scales best approximated by a modified measure theory~\cite{mmt-cfs}. This could open the door for a dark-matter-MOND duality. Modified Netwonian dynamics~\cite{milgrom1983modification} is a competing hypothesis for the explanation of the rotation curves of galaxies, where it has been very successful~\cite{gentile2011things,mcgaugh2012baryonic}. However, MOND has been struggling to explain the dynamics on inter-galactic and universal scales~\cite{conundrum,clowe2006direct,skordis2006large}. The \ac{fcdm} model would provide a natural explanation for this disparity. While the dark matter states can be assigned to the background on sub-galactic scales, they are well localized when we consider inter-galactic dynamics, and hence in this setting they are firmly within the matter sector of the effective description. This has the potential to unify the two hypotheses explaining the rotation curves of galaxies. It would be substantial progress if these ideas could be made rigorous. 
    
    \item The present paper was motivated largely by the work in~\cite{finster2024causal} where a collapse model was derived from \ac{cfs} in the non-relativistic setting. The interesting aspect with respect to the present discussion is that the stochastic fields are obtained as linearized perturbations of subsystems $(\H_i,\omega_i)$ of the Minkowski vacuum. This gives rise to a plethora of effective fields in the continuum description. To calculate the dynamics of a probe $\psi (x)\in \H$ we have to take all those fields into account for which $\psi(x) \in \H_i$, and we have to keep in mind that by virtue of being contained in $\H_i$, $\psi(x)$ itself contributes to the collective perturbation described by the field. This is a manifestation of the fundamental nonlinearity of the theory at the emergent level. However, one key question arising from the current work is the following. The derivation of the collapse model assumes an abstract observable $\mathcal{O}$ that commutes both with the non-relativistic Hamiltonian and with the stochastic potentials. To develop a truly comprehensive physical interpretation of the theory, it will be important to understand whether (some of) the observables formulated in the present paper satisfy these criteria. 

\item The above considerations are particularly interesting in the context of the connection between \ac{cfs} and Fröhlich's ETH approach to \ac{qt}~\cite{ethcfs,froehlich2019review,froehlich2019relativistic}. If we assume a \ac{cfs} given by a discrete measure $\rho$, then the observables we get are projection operators with $2n$ dimensional eigenspace. So we get a starting point compatible with Fröhlich's setup. Now, a loss of access to information in the present setting would intuitively correspond to a setting where the Hilbert space encoding the future of a point $\bigcup_{y \textit{ in future of } x} S_y$ is a strict subset of the Hilbert space $\H$. This implies that there exist states that do not contribute to the correlations in the future of $x$. This can also be understood as there existing states $v\neq u \in \H$ who's physical wave functions agree on the relevant subset of spacetime. In this case they carry the exact same information and are thus redundant signifying that the future encodes less information than the past\footnote{If the physical wave functions of two states agree in the entire spacetime $\psi_u(x)=\psi_v(x) \; \forall x\in M$ then one of them is redundant, i.e.~an element of the auxiliary Hilbert space. This is essentially Shannon's sampling theorem at work in a universe with finite information. Therefore, if we add physically relevant states to a \ac{cfs} the number of spacetime points in $M$ has to increase, giving us a finer resolution of the approximated continuum spacetime. Working with the auxiliary Hilbert space allows us to encode all this in the measure and the projection alone without having to change any other structure. }. Trying to make this connection more rigorous could help sharpen the concept of ``information'' in a \ac{cfs} which we have deployed frequently and rather loosely throughout the paper.

\item The way in which we introduced the fermionic projector in Section~\ref{sec:causalcorrelations} might allow us to establish a connection to Barandes reformulation of \ac{qt}~\cite{barandes2023stochasticquantumcorrespondence,barandes2024newprospectscausallylocal}. A priori, there is no reason why the information transfer operator should be unitary. Furthermore, it should be noted that the nonlocality in time in~\cite{finster2024causal} does not allow the evolution to be cut at any arbitrary point which is conceptually in agreement with the arguments of Barandes. Only at points where the wave function collapsed onto a particular observable can we restart the evolution. If the two sets of ideas could be aligned, this would allow us to tap into a whole other set of ideas to advance \ac{cfs} phenomenology.

\end{enumerate}

As a final remark, we comment on two insights from solid-state theory that might apply in the context of \ac{cfs} or quantum gravity more broadly. First, in solid-state theory the macroscopic properties of a material are largely dominated by the band structure around the Fermi energy. The rest of the band structure, no matter how intricate, does not really matter, except to determine where the Fermi energy comes to lie. In a similar way, in \ac{cfs} the cusp of the mass shell and the mass gap carry most of the information that characterizes the effective macroscopic description of the physical system. Second, only in extreme circumstances such as the formation
of Cooper pairs in superconductors does the electron-electron interaction change the dynamics of the system substantially. Therefore, the non-interacting picture is a good approximation for the vacuum spacetime in \ac{cfs}.

 \printbibliography    
\end{document}